\documentclass[a4paper,10pt,noparskip]{article}
\usepackage[a4paper,biblatex]{style}
\usepackage{xurl}
\newcommand{\TV}{d_{\mathrm{TV}}}
\newcommand{\fall}[2]{#1^{\downarrow #2}}

\begin{document}
\pagenumbering{roman}

\title{Optimal spectrum estimation}

\author{
Ainesh Bakshi\footnote{NYU. \texttt{ainesh@nyu.edu}} 
\and
Apoorv Vikram Singh\footnote{NYU. \texttt{apoorv.singh@nyu.edu}}
\and
Xinyu Tan\footnote{MIT. \texttt{norahtan@mit.edu}}
}
\date{}

\maketitle
\thispagestyle{empty}

\begin{abstract}
We prove that the spectrum of an unknown $d$-dimensional quantum state can be
estimated to error $\varepsilon$ in total variation distance using
\[
 O\!\left(d^2\min\left\{
 \frac{1}{(\varepsilon\log d)^4},\;
 \frac{1}{(\varepsilon\log d)^2}
 \right\}\right)
\]
copies. This matches the recent lower bound of Wang~\cite{wang2026nearlytightlowerbounds}. 
When restricted to unentangled measurements, we give an algorithm with an additional factor of $d$ in copy complexity, 
which we conjecture to be optimal. 

We develop a framework for recovering the small eigenvalues of a quantum state by matching Chebyshev moments.
We bound the variance of each Chebyshev moment estimate in terms of scalar derivatives of the corresponding polynomial, using classical and quantum Efron--Stein decompositions. 
Different rescalings of the Chebyshev polynomials balance approximation error and variance, yielding two regimes in our copy complexity bound.
\end{abstract}

\clearpage
\microtypesetup{protrusion=false}
\tableofcontents
\thispagestyle{empty}
\microtypesetup{protrusion=true}
\clearpage
\pagenumbering{arabic}
\section{Introduction}

The spectrum of a quantum state captures all the information that is invariant under a
change of basis. In particular, it provides access to fundamental quantities such as purity,
von Neumann entropy, and, more generally, R\'enyi entropies, making it an
indispensable tool for studying quantum states and processes~\cite{HM02}. In many-body physics, the spectrum of a reduced
density matrix is a basic tool for quantifying entanglement: for a bipartite pure
state, its eigenvalues are the squared Schmidt coefficients and determine the
entanglement entropy~\cite{BHACRG18}. In quantum
chemistry, the entropies of reduced density matrices are used to characterize
electron correlation and understand the formation and breaking of chemical
bonds~\cite{BTBLR13}. Understanding how to estimate these spectra, therefore,
is a fundamental question that cuts across quantum information, many-body
physics, and quantum chemistry.

In the spectrum estimation problem, we are given copies of an unknown density
matrix
$\rho\in\mathbb{C}^{d\times d}$ with eigenvalues
$\alpha_1\geq\cdots\geq\alpha_d$. The goal is to output a sorted probability vector
$\widehat{\bm\alpha}$ whose total variation distance from $\alpha$ is at most
$\varepsilon$, using the minimum number of copies. Full state tomography
provides an immediate way to solve this problem and the number of copies needed for tomography
depends on the measurements available. With \emph{entangled measurements},
which act jointly on all copies, the optimal copy complexity is
$\Theta(d^2/\varepsilon^2)$~\cite{OW16, HHJWY17}. With \emph{unentangled
measurements}, which act on one copy at a time and may be chosen adaptively,
the optimal complexity is $\Theta(d^3/\varepsilon^2)$~\cite{CHLLS23}.
The key challenge in spectrum estimation is to determine how much of the cost
of learning the eigenbasis can be avoided. If the eigenbasis were known,
measurements in that basis would reduce the problem to classical
\emph{sorted distribution estimation}. For constant accuracy, this requires only $\Theta(d/\log d)$ samples~\cite{VV17,HJW18}.

In his 2016 thesis, Wright~\cite[Section~10.2]{Wright16} conjectured that
spectrum estimation requires $\Omega(d^2/\log d)$ copies at constant accuracy,
allowing at most a logarithmic improvement over full state tomography.
Recent work has demonstrated that spectrum estimation is indeed easier than
full state tomography. We summarize the known upper and lower bounds in
\cref{tab:prior-spectrum-bounds}.

\begin{table}[ht]
\centering
\small
\setlength{\tabcolsep}{6pt}
\setlength{\bigstrutjot}{30pt}
\begin{tabular}{|>{\centering\arraybackslash}p{0.14\linewidth}|>{\centering\arraybackslash}p{0.34\linewidth}|p{\dimexpr0.52\linewidth-6\tabcolsep-4\arrayrulewidth\relax}|}
\hline
\rule{0pt}{13pt} & \textbf{Upper bounds} & \multicolumn{1}{c|}{\textbf{Lower bounds}} \\[3pt]
\hline
\parbox[c][40pt][c]{\linewidth}{\centering Unentangled}
& \parbox[c][40pt][c]{\linewidth}{\centering
  $\displaystyle O\!\left(\frac{d^3}{\varepsilon^6}
    \left(\frac{\log\log d}{\log d}\right)^4\right)$~\cite{PTTW26}}
& \multirow[t]{2}[2]{=}[10pt]{
  $\Omega(d^{2-\gamma})$~\cite{FOW26}\newline
  Constant accuracy; any fixed $\gamma>0$.
  \par\medskip
  $\Omega(d^2/(\log d)^c)$~\cite{LJ26}\newline
  Constant accuracy; some fixed constant $c$.
  \par\medskip
  $\Omega\!\left(d^2\min\{(\varepsilon\log d)^{-4},(\varepsilon\log d)^{-2}\}\right)$~\cite{wang2026nearlytightlowerbounds}
  } \\
\cline{1-2}
\parbox[c][36pt][c]{\linewidth}{\centering Entangled}
& \parbox[c][36pt][c]{\linewidth}{\centering
  $\displaystyle O\!\left(\frac{d^2}{\varepsilon^4}
    \left(\frac{\log\log d}{\log d}\right)^2\right)$~\cite{PSTW26}}
& \\
\hline
\end{tabular}
\caption{Prior works on the copy complexity of spectrum estimation.
All three lower bounds hold for any entangled measurements and hence apply to the weaker model of unentangled measurements.}
\label{tab:prior-spectrum-bounds}
\end{table}

Even for constant $\varepsilon$, the optimal dependence of the copy complexity
on $d$ remains unknown. More generally, the existing bounds do not determine
how the optimal rate changes as $\varepsilon$ decreases, or whether distinct
accuracy regimes exhibit different scaling. Pelecanos, Spilecki, Tang, and
Wright~\cite[Section~1.3]{PSTW26} conjectured that the optimal copy complexity
with entangled measurements is $\Theta(d^2/(\varepsilon^2\log^2 d))$, which leaves open the possibility of an improved lower bound in the large $\varepsilon$ regime. Therefore, the central question we address in this work is as follows:

\begin{quote}
\begin{center}
\emph{What are the optimal rates for learning
the spectrum of a quantum state?}
\end{center}
\end{quote}

\subsection{Our results}

We determine the optimal copy complexity of spectrum estimation
with entangled measurements. Formally, 

\begin{theorem}[Entangled spectrum estimation]
\label{thm:entangled}
Given $0<\varepsilon<1$ and an unknown quantum state $\rho\in\mathbb{C}^{d\times d}$, there is a quantum algorithm that outputs a sorted probability vector $\widehat{\bm\alpha}$ satisfying $\TV(\widehat{\bm\alpha}, \, \spec(\rho))\leq\varepsilon$ with probability at least $0.99$ using
\[
 O\!\left(d^2\min\left\{
 \frac{1}{(\varepsilon\log d)^4},\;
 \frac{1}{(\varepsilon\log d)^2}
 \right\}\right)
\]
copies of $\rho$.
\end{theorem}

\begin{remark}Our algorithm improves on the copy complexity of~\cite{PSTW26} throughout
the full range of accuracy parameters.
In particular, at constant accuracy, our algorithm uses $O(d^2/\log^4 d)$
copies. This refutes their conjectured rate of $\Theta(d^2/(\varepsilon^2\log^2 d))$. 
Moreover, our bound matches the lower bound of
Wang~\cite[Theorem~5.15]{wang2026nearlytightlowerbounds} up to constant factors. 
\end{remark}

If the unknown quantum state $\rho$ is promised to have rank at most $r$, then $d$ can be replaced by $r$ in the copy complexity bound in \Cref{thm:entangled}. This is achieved by first applying the random dimension reduction described in Lowe and Tan \cite{LT26} to $\rho^{\otimes n}$, which returns $n$ copies of an $r$-dimensional random state which has the same nonzero eigenvalues as $\rho$: $(\bU \diag(\alpha_1, \ldots, \alpha_r)\, \bU^\dagger)^{\otimes n}$, where $\bU$ is a Haar random $r \times r$ unitary. We then apply the same spectrum estimation algorithm but on these $r$-dimensional copies. 
Wang's lower bound also applies to quantum states supported on a fixed $r$-dimensional subspace. Therefore we also determine the optimal copy complexity under this rank-$r$ constraint. 

We also obtain an analogous guarantee for unentangled measurements,
with an additional factor of $d$ in the number of copies. Formally,

\begin{theorem}[Unentangled spectrum estimation]
\label{thm:unentangled}
Given $0<\varepsilon<1$ and an unknown quantum state $\rho\in\mathbb{C}^{d\times d}$, there is a quantum algorithm that outputs a sorted probability vector $\widehat{\bm\alpha}$ satisfying $\TV(\widehat{\bm\alpha}, \, \spec(\rho))\leq\varepsilon$ with probability at least $0.99$ using
\[
 O\!\left(d^3\min\left\{
 \frac{1}{(\varepsilon\log d)^4},\;
 \frac{1}{(\varepsilon\log d)^2}
 \right\}\right)
\]
copies of $\rho$ and unentangled measurements only. 
\end{theorem}

\begin{remark}
At constant accuracy, this gives a copy complexity of $O(d^3/\log^4 d)$,
removing the $(\log\log d)^4$ factor in the bound of Pelecanos, Tan, Tang,
and Wright~\cite{PTTW26}. The theorem also improves the dependence on
accuracy from $\varepsilon^{-6}$ to $\varepsilon^{-4}$ or
$\varepsilon^{-2}$, depending on the regime. We conjecture that this copy complexity
is optimal among all algorithms using possibly adaptive unentangled measurements. 
\end{remark}

Our algorithms for entangled and unentangled spectrum estimation follow a common two-step template. 
First, we estimate the large eigenvalues
and construct a projector that approximately separates the large and small
spectral components. Then we estimate the small eigenvalues by matching moments. This two-step template is the same as the one used in \cite{PSTW26,PTTW26}, and the first step is implemented using the \emph{bucketing algorithms} developed there. 
However, our analysis of the second step is substantially different. 
We introduce a framework for analyzing moment matching based on
Chebyshev polynomials, inspired by the work of Musco, Musco, Rosenblatt, and
Singh~\cite{MMRS25}. We show how to reconstruct the spectrum from approximate
Chebyshev moments, i.e., we control how errors in the estimated
moments translate into errors in the recovered eigenvalues. 

The Chebyshev moments are obtained as a linear combination of the underlying monomial moment estimates. 
To have good bounds on the variance of the estimated Chebyshev moments, we use the classical and quantum Efron--Stein decompositions to reduce the variance analysis to bounds on matrix derivatives. 
We then control these matrix derivatives using scalar derivatives of the corresponding Chebyshev polynomials. 
Smaller variance bounds allow us to use polynomials of higher degree in reconstructing the spectrum. 
The two rates arise from different rescalings of the Chebyshev polynomials
that balance approximation error against variance. 
For the first
rate, we map the spectrum to $[-1,1]$ and use a polynomial approximation
whose error improves for eigenvalues near zero. Since the eigenvalues sum
to at most one, most must be small, so this finer approximation gives a
sharper bound on the total error summed over the spectrum. This yields the
$(\varepsilon\log d)^{-4}$ dependence, albeit any such approximation has large derivatives near the endpoints. To improve the dependence on
$\varepsilon$, we instead map the spectrum to $[-1/2,1/2]$, avoiding the
large derivatives at the endpoints. The resulting reduction in variance
allows us to use higher degree polynomials, compensating for the weaker
approximation guarantee and yielding the $(\varepsilon\log d)^{-2}$
dependence.

\section{Technical overview}
\label{sec:technical-overview}

In this section, we explain the proof for entangled spectrum estimation, and the analysis for the unentangled case is very similar. 
We start by recalling two ingredients from \cite{PSTW26} that we use as black boxes. 

\paragraph{Bucketing and estimating large eigenvalues.}
The first ingredient is a \emph{bucketing algorithm} which, given a threshold $0<B<1$, uses $O(d/(B\varepsilon^2))$ copies to
estimate the eigenvalues of $\rho$ greater than $B$ and construct a corresponding large-bucket projector $\Pi$ of rank $O(1/B)$. $\overline{\Pi}=I-\Pi$ is called the \emph{small-bucket projector} and 
$\sigma=\overline{\Pi}\rho\overline{\Pi}$ is called the \emph{small-bucket state}. 
Set $L=1.1B$. With probability at least $0.99$, the large eigenvalues can be estimated to $O(\varepsilon)$ TV error, 
\begin{equation*}
 0\preceq\sigma\preceq LI_d, \qquad\text{and}\qquad \TV(\spec(\rho), \, \spec(\Pi\rho\Pi+\sigma)) \leq O(\epsilon).  
\end{equation*}
Throughout the overview, 
we condition on a fixed successful output from this bucketing step. 
Therefore, the problem is reduced to estimating the spectrum of the small-bucket state $\sigma$, denoted by $L\geq \alpha_1\geq \cdots \geq \alpha_d\geq 0$, with TV error $O(\varepsilon)$.

\paragraph{Measurements on the small-bucket state.}
The second ingredient is a fully entangled measurement
that provides unbiased estimates of the monomial moments of $\sigma$.
On $n$ fresh copies of $\rho$, we measure $\{\Pi,\overline{\Pi}\}$ and apply
\emph{weak Schur sampling} jointly to the copies for which the outcome
is $\overline{\Pi}$. 
Denote the measurement outcome by $\blambda$. Plugging  $\blambda$ into the formula in \cite[Definition~5.5]{PSTW26} gives real estimates
$\widehat{\bm M}_1,\ldots,\widehat{\bm M}_K$, for any $K\leq n$, such that
$\E\widehat{\bm M}_m=\tr(\sigma^m)$.  

Therefore, for a univariate 
polynomial $p(x)=\sum_{m=0}^Kp_mx^m$, we can construct an unbiased estimate
of $\tr(p(\sigma))$, i.e., the sum of $p$ evaluated at the eigenvalues of $\sigma$,  as follows:
\begin{equation}
\label{eq:overview-polynomial-estimator}
 \widehat{\bm F}_p=p_0d+\sum_{m=1}^Kp_m\widehat{\bm M}_m,
 \qquad
 \E\widehat{\bm F}_p=\sum_{i=1}^dp(\alpha_i)=\tr(p(\sigma)).
\end{equation}
Observe that all of these estimates are computed from the same measurement outcome $\blambda$, and thus their errors can be correlated.

\vspace{1em}

Our improvements concern the reconstruction of the small eigenvalues from these moment estimates.
We first bound the spectrum recovery error in terms of polynomial approximation error and moment estimation variances. 
We then bound these variances using polynomial derivatives and choose the polynomials to obtain the two copy complexity bounds.

\subsection{Chebyshev moment matching}\label{sec:overview-cheb-matching}

It remains to understand how the statistics, $\widehat{\bm F}_p$, determine the true spectrum $\alpha$. We begin by considering the set of feasible candidates for the spectrum of $\sigma$:
\begin{equation*}
 \mathcal Z_d
 =\braces[\Big]{z\in[0,L]^d:z_1\geq\cdots\geq z_d,
                    \ \sum_i z_i\leq1}.
\end{equation*}
Let $\mu_\alpha=\sum_i\delta_{\alpha_i}$ and
$\mu_z=\sum_i\delta_{z_i}$ be the counting measures associated with $\alpha$ and $z$, where each measure places one unit of mass at every (potentially zero) eigenvalue. Since $\alpha$ and $z$ are sorted, the optimal transport plan is to match the points in sorted order and $\|\alpha-z\|_1=2\TV(\alpha,z)$ is exactly the Wasserstein distance between the two measures. By Kantorovich-Rubinstein duality,
\begin{equation}
\label{eq:overview-wasserstein-duality}
 2\TV(\alpha,z)
 =\sup_{\operatorname{Lip}(f)\leq1}
   \left|\sum_i f(\alpha_i)-\sum_i f(z_i)\right|.
\end{equation}
Thus it suffices to control the difference between the two spectra
when tested against any $1$-Lipschitz function. A natural approach
is to approximate $f$ by a low degree polynomial
$P_f(x)=\sum_{m=0}^K b_mx^m$. Adding and subtracting the sums of
$P_f$ over the two spectra gives
\begin{equation*}
 \left|\sum_i f(\alpha_i)-\sum_i f(z_i)\right|
 \leq
 \underbrace{\sum_i\bigl(|f(\alpha_i)-P_f(\alpha_i)| + |f(z_i)-P_f(z_i)|\bigr)}_{\text{approximation error}}
 + \underbrace{\left|\sum_{m=1}^K b_m
   \left(\sum_i\alpha_i^m-\sum_i z_i^m\right)\right|}
   _{\text{moment discrepancy}}.
\end{equation*}
The first term is the approximation error summed over both the true and candidate spectra.
The second is a linear combination of their moment differences,
weighted by the coefficients of $P_f$. We will eventually choose $z$ to fit
the estimated moments. For each $m$, the moment difference can then be split as follows:
\begin{equation}
\label{eq:overview-moment-errors}
 \sum_i\alpha_i^m-\sum_i z_i^m
 =\underbrace{\sum_i\alpha_i^m-\widehat{\bm M}_m}_{\text{estimation error}}
 +\underbrace{\widehat{\bm M}_m-\sum_i z_i^m}_{\text{fitting error}}.
\end{equation}
We therefore need to control both the quality of the approximation
and how its coefficients amplify errors in the moments. Crucially, we note that the choice of basis matters when the moments are noisy.
The monomial coefficients $b_m$ can be exponentially large
in the degree, even for a polynomial bounded on the interval. For
instance, the degree-$k$ Chebyshev polynomial $T_k(2x/L-1)$ is bounded
by one on $[0,L]$, but its leading monomial coefficient is
$2^{2k-1}L^{-k}$. Bounding the contributions of the monomials
separately can therefore lose the cancellations that
keep the polynomial bounded.  

\paragraph{Polynomial approximation on the spectral interval.}
We begin with Jackson's theorem, which states that for every
$\Lambda$-Lipschitz function $g:[-1,1]\to\mathbb R$ and integer $K\geq1$,
there is a polynomial $R$ of degree at most $K$ such that $$\sup_{t\in[-1,1]}|g(t)-R(t)|
 \lesssim\frac{\Lambda}K.$$
For a $1$-Lipschitz function
$f:[0,L]\to\mathbb R$, the rescaled function
$x\mapsto f(L(1+x)/2)$ is $L/2$-Lipschitz. Applying Jackson's
theorem and rescaling back to $[0,L]$ therefore gives a polynomial
approximation to $f$ with error $O(L/K)$.
Summing this error over the eigenvalues gives a total approximation
error of $O(dL/K)$. However, observe the eigenvalues sum to at most one, 
so we should instead use an approximation whose error improves near zero. 
This strengthening of Jackson's theorem is due to
DeVore~\cite{DeVore76} (a similar observation was made in the classical setting~\cite[Lemma~22]{HJW18}) and
states that for every $1$-Lipschitz function
$f:[0,L]\to\mathbb R$ and integer $K\geq1$, there is a polynomial
$P_f$ of degree at most $K$ such that
\begin{equation*}
 |P_f(x)-f(x)|\lesssim \frac{\sqrt{Lx}}K
 \qquad\text{for all }x\in[0,L].
\end{equation*}
Notice, this bound recovers Jackson's theorem when $x\gtrsim L$. Summing over the eigenvalues and applying Cauchy--Schwarz gives
\begin{equation}
  \label{eqn:overview-approx-error}
 \sum_i|P_f(\alpha_i)-f(\alpha_i)|
 \lesssim\frac{\sqrt L}{K}\sum_i\sqrt{\alpha_i}
 \leq\frac{\sqrt L}{K}\sqrt{d\sum_i\alpha_i}
 \leq\frac{\sqrt{dL}}K,
\end{equation}
where the last inequality uses $\sum_i\alpha_i\leq1$.  
It remains to control the coefficients of this approximation, which
govern how errors in the moments affect the reconstruction. We do
this by expressing $P_f$ in the Chebyshev basis. We let $T_k$ be the Chebyshev polynomial (\Cref{def:chebyshev-poly}), and consider the shifted and
scaled polynomials
\begin{equation*}
 \phi_k(x)=T_k(2x/L-1)-(-1)^k,
 \qquad 1\leq k\leq K.
\end{equation*}
The rescaling maps $[0,L]$ onto $[-1,1]$, and the subtracted
constant ensures that $\phi_k(0)=0$. Since constants cancel
in~\eqref{eq:overview-wasserstein-duality}, we may assume $f(0)=0$,
which also gives $P_f(0)=0$. These polynomials form a basis
for the polynomials of degree at most $K$ that vanish at zero, so
we can write $P_f=\sum_{k=1}^Ka_k\phi_k$. For polynomials expressed in the Chebyshev basis together with
the uniform approximation guarantee $\|P_f-f\|_\infty\lesssim L/K$, Musco, Musco, Rosenblatt and Singh~\cite[Lemma~13]{MMRS25} provide a coefficient bound of  $\sum_{k=1}^K k^2a_k^2\lesssim L^2$.
This bound controls both the estimation and fitting errors
in~\eqref{eq:overview-moment-errors}. Indeed, for any error vector
$u\in\mathbb R^K$, Cauchy--Schwarz implies
\begin{equation}
  \label{eqn:coefficient-bound-in-action}
 \left|\sum_{k=1}^K a_ku_k\right|
 \leq\left(\sum_{k=1}^K k^2a_k^2\right)^{1/2}
      \left(\sum_{k=1}^K\frac{u_k^2}{k^2}\right)^{1/2}
 \lesssim L\left(\sum_{k=1}^K\frac{u_k^2}{k^2}\right)^{1/2}.
\end{equation}
Further, it is known that such a coefficient bound is as good as it gets, i.e., \cite{MMRS25} show a $1$-Lipschitz function that achieves the coefficient bound exactly. Using the coefficient bounds of Musco, Musco, Rosenblatt, and Singh as a black box only yields a copy complexity of $O(d^2 (\varepsilon \log(d))^{-4})$. To obtain the improved dependence on $\varepsilon$, we need a sharper bound. The key insight here is that~\cite{MMRS25}'s bound only uses that the test function $f$ is $1$-Lipschitz, but to certify the distance
between two spectra, it suffices to use a function that attains the
supremum in~\eqref{eq:overview-wasserstein-duality} and such functions have more structure. In particular,  let $F(x)$ be the difference between
the numbers of true and candidate eigenvalues at most $x$, i.e.,
\begin{equation*}
  F(x)=\#\{i:\alpha_i\leq x\}-\#\{i:z_i\leq x\}.
\end{equation*}
Under the sorted matching, $|F(x)|$ counts the pairs separated by the
threshold $x$, so integrating $|F(x)|$ gives the distance between
the two spectra. We therefore choose $f$ to be constant where
the counts agree, and to have slope $-\operatorname{sgn}(F(x))$
where they differ. This choice allows us to control the integral of
the squared derivative of $f$: since each eigenvalue contributes
one unit to the counting measure, whenever $F(x)$ is nonzero,
its magnitude is at least one. Therefore,
\begin{equation*}
  \int_0^L |f'(x)|^2\,dx
 =\int_0^L \mathbf1_{\{F(x)\ne0\}}\,dx
 \leq\int_0^L |F(x)|\,dx
 =2\TV(\alpha,z).
\end{equation*}
We then prove the following sharper coefficient bound for polynomials approximating a Lipschitz test function $f$ which scales with the derivative of $f$: 
\begin{equation}
\label{eq:overview-derivative-energy}
 \sum_{k=1}^K k^2a_k^2
 \lesssim L\int_0^L|f'(x)|^2\,dx.
\end{equation}
The integral is at most $L$ for any $1$-Lipschitz function, which
recovers the~\cite{MMRS25} bound. However, for our particular choice of $f$,
it is bounded by $2\TV(\alpha,z)$.
Therefore, we can strengthen
\cref{eqn:coefficient-bound-in-action} to
\begin{equation}
\label{eq:overview-refined-coefficient-bound}
 \left|\sum_{k=1}^K a_ku_k\right|
 \lesssim\sqrt{L\TV(\alpha,z)}
 \left(\sum_{k=1}^K\frac{u_k^2}{k^2}\right)^{1/2}
 \qquad\text{for every }u\in\mathbb R^K.
\end{equation}

\paragraph{Reconstructing the spectrum.}
We now use~\cref{eq:overview-refined-coefficient-bound} to control
the error in the reconstructed spectrum. By integration by parts,
our choice of $f$ satisfies
$\sum_i f(\alpha_i)-\sum_i f(z_i)=2\TV(\alpha,z)$.
As in~\eqref{eq:overview-moment-errors}, we split each Chebyshev
moment discrepancy into an estimation error and a fitting error,
and apply the refined coefficient bound to each contribution.
The approximation bound in~\eqref{eqn:overview-approx-error}
holds for both spectra, since each has sum at most one. Combining
these bounds gives
\begin{equation}
\label{eq:overview-deterministic-recovery}
 \begin{aligned}
 \TV(\alpha,z)
 \lesssim\frac{\sqrt{dL}}K
      &+\sqrt{L\TV(\alpha,z)}\,
      \underbrace{\left(\sum_{k=1}^K\frac1{k^2}
        \left(\sum_i\phi_k(\alpha_i)-\widehat{\bm F}_{\phi_k}\right)^2
        \right)^{1/2}}_{\text{estimation error}}\\
 &+\sqrt{L\TV(\alpha,z)}\,
      \underbrace{\left(\sum_{k=1}^K\frac1{k^2}
        \left(\widehat{\bm F}_{\phi_k}-\sum_i\phi_k(z_i)\right)^2
        \right)^{1/2}}_{\text{fitting error}}.
 \end{aligned}
\end{equation}

The fitting error suggests a natural reconstruction algorithm:
choose the feasible spectrum that minimizes the weighted squared
discrepancy from the estimated Chebyshev moments. More formally,
given $\widehat{\bm F}_{\phi_1},\ldots,\widehat{\bm F}_{\phi_K}$, let
\begin{equation*}
 \widehat{\bm z}\in\argmin_{z\in\mathcal Z_d}
 \sum_{k=1}^K\frac1{k^2}
 \left(\sum_i\phi_k(z_i)-\widehat{\bm F}_{\phi_k}\right)^2.
\end{equation*}
Since the true spectrum $\alpha$ is feasible, the fitting error of
$\widehat{\bm z}$ is at most the estimation error. Substituting into
\eqref{eq:overview-deterministic-recovery}, both contributions are
therefore bounded by the weighted estimation error, with a coefficient of 
$\sqrt{L\TV(\alpha,\widehat{\bm z})}$. Using $ab \leq a^2/2 + b^2/2$, we can rearrange \cref{eq:overview-deterministic-recovery} to obtain
\begin{equation*}
 \TV(\widehat{\bm z},\alpha)
 \lesssim\frac{\sqrt{dL}}K
 +L\sum_{k=1}^K\frac1{k^2}
       \left(\sum_i\phi_k(\alpha_i)-\widehat{\bm F}_{\phi_k}\right)^2.
\end{equation*}
Taking expectations and recalling that each $\widehat{\bm F}_{\phi_k}$ is
unbiased, we have that the expected squared estimation error is precisely
its variance, yielding
\begin{equation}
\label{eq:overview-recovery}
 \E\TV(\widehat{\bm z},\alpha)
 \lesssim\frac{\sqrt{dL}}K
 +L\sum_{k=1}^K\frac{\Var(\widehat{\bm F}_{\phi_k})}{k^2}.
\end{equation}
Note, the estimates here need not be independent (in fact they are not), and we have reduced our analysis to bounding the variances of these correlated Chebyshev moment estimates.  
The estimation term is now linear in the variances of the
Chebyshev moment estimates. Taking expectations here uses only
unbiasedness; the estimates need not be independent. It remains
to bound these variances in terms of the number of copies.

\subsection{Variance bounds from polynomial derivatives}

We now bound the variance of $\widehat{\bm F}_p$ for a real
polynomial $p(x)=\sum_{m=0}^Kp_mx^m$.
A natural approach, used in~\cite{PTTW26,PSTW26,HJW18}, is to bound the errors in the monomial moments
separately, and then combine these bounds using the coefficients
of the polynomial. By the triangle inequality, 
\begin{equation*}
 \sqrt{\Var(\widehat{\bm F}_p)}
 \leq\sum_{m=1}^K|p_m|\sqrt{\Var(\widehat{\bm M}_m)}.
\end{equation*}
However, the monomial coefficients can be exponentially large in
the degree, even when the polynomial is bounded on the
interval $[0,L]$. The individual moment bounds also introduce
factors of $K^{O(K)}$. With these estimates, the analyses
in~\cite{PSTW26,PTTW26} take $K=\Theta(\log d/\log\log d)$.
Our improvement is to retain the cancellations between the
correlated moment errors. We instead bound the variance of $\widehat{\bm F}_p$ directly, and show
that it is controlled by the derivatives of the \emph{univariate} polynomial $p$ on $[0,L]$.

\paragraph{Moment estimates as observables.}
To analyze this variance, we first express the estimates $\widehat{\bM}_m$ of the monomial moments as measurements of some observables $\mathsf{M}_m$ on $\rho^{\otimes n}$. 
This then establishes the estimates $\widehat{\bF}_p$ as measurements of observables of the form $\mathsf{F}_p = \sum_{m=0}^K p_m \mathsf{M}_m$ on $\rho^{\otimes n}$.
Let $\overline\Pi_i$ denote $\overline\Pi$ acting on copy $i$, and let $R_{(i_1\cdots i_m)}$
cyclically permute the indicated copies. Let $n^{\downarrow m}=n(n-1)\cdots(n-m+1)$ be the downward factorial, $\mathsf M_0:=dI$ and 
\begin{equation*}
 \mathsf M_m:=\frac1{n^{\downarrow m}}
 \sum_{\substack{i_1,\ldots,i_m\in[n]\\\text{all distinct}}}
 \overline\Pi_{i_1}\cdots\overline\Pi_{i_m}R_{(i_1\cdots i_m)},
 \qquad 1\leq m\leq K. 
\end{equation*}
The operators $\mathsf M_m$ are Hermitian and invariant under permutations of the $n$ copies.
For any density matrix $\tau$, the expectation of measuring each observable $\mathsf{M}_m$ on $\tau^{\otimes n}$ is 
\begin{equation*}
  \tr(\mathsf{M}_m \tau^{\otimes n}) = \tr((\overline{\Pi} \tau \overline{\Pi})^m),
\end{equation*}
which is the $m$-th monomial moment of the projected state $\overline{\Pi} \tau \overline{\Pi}$. 
We note that these moment observables $\mathsf M_1,\ldots,\mathsf M_K$ can be measured simultaneously by the procedure described at the beginning of the overview: 
measure $\{\Pi,\overline\Pi\}$ on each copy and apply weak Schur sampling to the retained copies.
Plugging the measurement outcome $\blambda$ into the formula in~\cite[Definition~5.5]{PSTW26} gives their measured values $\widehat{\bm M}_1,\ldots,\widehat{\bm M}_K$. Therefore the expected value of measuring $\mathsf F_p\coloneqq\sum_{m=0}^Kp_m\mathsf M_m$ on $\tau^{\otimes n}$ is 
\begin{equation}
\label{eq:overview-observable-mean}
 \tr(\tau^{\otimes n}\mathsf F_p)
 =\tr\!\left(p(\overline\Pi\tau\overline\Pi)\right). 
\end{equation}
Thus $\widehat{\bm F}_p = p_0d + \sum_{m=1}^K p_m \widehat{\bM}_m$ given in~\cref{eq:overview-polynomial-estimator} is the measured value of $\mathsf F_p$ on $\rho^{\otimes n}$, and its variance is given by
\begin{equation*}
  \Var(\widehat{\bm F}_p) = \tr(\rho^{\otimes n}\mathsf F_p^2) - \tr(\rho^{\otimes n}\mathsf F_p)^2 . 
\end{equation*}

\paragraph{Quantum Efron--Stein decomposition.}
The main difficulty in bounding the variance of $\widehat{\bF}_p$ is
controlling the second moment $\tr(\rho^{\otimes n}\mathsf F_p^2)$.
Expanding the square and bounding the resulting terms separately can
lose the cancellations between moment estimation errors. We instead
use the quantum Efron--Stein decomposition to write the centered
observable as a sum of operators acting on subsets of the input
copies. These operators are orthogonal with respect to
$\rho^{\otimes n}$, so their second moments add to give the variance.
We then identify these operators through derivatives of the mean,
for which we already have the polynomial expression
in~\cref{eq:overview-observable-mean}. This will reduce the variance
bound to controlling operators on at most $K$ copies, where $K$ is
the degree of $p$.

We briefly recall the quantum Efron--Stein decomposition~\cite{GB10,PFMO25}.
For any observable $G$ on $n$ copies, we can write
\begin{equation*}
 G=\tr(G\rho^{\otimes n})I^{\otimes n}
 +\sum_{\varnothing\ne S\subseteq[n]}Z_S,
\end{equation*}
where $Z_S$ acts only on the copies in $S$ and is centered in each
of these copies: averaging any one of them against $\rho$ gives
the zero operator. This centering property implies that distinct
components are orthogonal. Indeed, if $S\ne T$, there is a copy
belonging to exactly one of the two subsets. On this copy, one
operator is centered and the other acts as the identity, so
averaging gives $\tr(\rho^{\otimes n}Z_SZ_T)=0$. Thus, expanding
the square of the centered observable and taking its expectation,
all cross terms vanish and we obtain
\begin{equation*}
 \tr(\rho^{\otimes n}G^2)-\tr(\rho^{\otimes n}G)^2
 =\sum_{\varnothing\ne S\subseteq[n]}
   \tr(\rho^{\otimes n}Z_S^2).
\end{equation*}

We now apply this decomposition to $G=\mathsf F_p$. Since
$\mathsf F_p$ is invariant under permutations of the copies,
every subset $S$ of size $r$ carries the same centered operator,
which we denote by $G_r$. The copies outside $S$ contribute a
factor of $\tr(\rho)=1$, so
$\tr(\rho^{\otimes n}Z_S^2)=\tr(\rho^{\otimes r}G_r^2)$.
There are $\binom nr$ subsets of size $r$, and grouping their
contributions gives
\begin{equation}
\label{eq:qES}
 \Var(\widehat{\bF}_p)
 =\sum_{r=1}^n\binom nr\tr(\rho^{\otimes r}G_r^2).
\end{equation}
Thus, it suffices to bound the second moments of the operators
$G_r$. We will do this using derivatives of the mean
in~\cref{eq:overview-observable-mean}, without computing $G_r$
explicitly. The key observation is that, when we perturb the input
state from $\rho$ to $\rho+t\Delta$, the centering property ensures
that an $r$-copy component contributes only to the coefficient
of $t^r$. This will allow us to relate $G_r$ to the $r$th
derivative of the mean.

\paragraph{Identifying the components via Taylor expansion.}
Let $\Delta=\tau-\rho$ for a density matrix $\tau$, and consider
the perturbed state $\rho+t\Delta$ for $t\in[0,1]$. Set
$H=\overline\Pi\Delta\overline\Pi$ and define
$g_H(t)=\tr(p(\sigma+tH))$. By~\cref{eq:overview-observable-mean},
this is the mean of our estimator on the perturbed input.
Fix a subset $S$ of size $r$ and observe that since $Z_S$ acts as $G_r$ on $S$
and as the identity elsewhere, we have
\begin{equation*}
 \tr\!\left[Z_S(\rho+t\Delta)^{\otimes n}\right]
 =\tr\!\left[G_r(\rho+t\Delta)^{\otimes r}\right]
 =t^r\tr(G_r\Delta^{\otimes r}).
\end{equation*}
The second equality follows by expanding the tensor product. Each
term containing $\rho$ in at least one factor vanishes by the
centering property of $G_r$, leaving only the term with $\Delta$
in all $r$ factors. Thus, an $r$-copy component contributes only
to the coefficient of $t^r$. Taking the expectation of the
decomposition and grouping subsets by their size gives
\begin{equation}
  \label{eq:mean-perturbation}
 g_H(t)=g_H(0)+\sum_{r=1}^n\binom nr t^r
   \tr(G_r\Delta^{\otimes r}).
\end{equation}
On the other hand, with $\sigma$ and $H$ fixed,
$g_H(t)=\tr(p(\sigma+tH))$ is a scalar polynomial in $t$ of
degree at most $K$. Taylor expanding this polynomial around
$t=0$ yields yet another expression for $g_H(t)$, namely $g_H(t)=g_H(0)+\sum_{r=1}^K\frac{g_H^{(r)}(0)}{r!}\,t^r$. 
Comparing the $r$-th coefficient with that in \cref{eq:mean-perturbation}, we can conclude that 
\begin{equation}
\label{eq:relate_derivative}
 \tr(G_r\Delta^{\otimes r})
 =\frac{g_H^{(r)}(0)}{n^{\downarrow r}},
 \qquad 1\leq r\leq K.
\end{equation}
Since $g_H$ has degree at most $K$, the same comparison gives
$\tr(G_r\Delta^{\otimes r})=0$ for $r>K$.

\paragraph{Matrix derivative operators.}
We now use~\cref{eq:relate_derivative} to bound the second moments
of the operators $G_r$. As a function of $H$, the derivative
$g_H^{(r)}(0)$ is a homogeneous polynomial of degree $r$ in its
entries. Since the entries of $H^{\otimes r}$ contain all products
of $r$ entries of $H$, we can collect their coefficients into an
operator $\Theta_p^{(r)}$ satisfying
\begin{equation*}
 g_H^{(r)}(0)=\tr\!\left(\Theta_p^{(r)}H^{\otimes r}\right)
 \qquad\text{for every Hermitian }H.
\end{equation*}
We can choose $\Theta_p^{(r)}$ to be Hermitian and invariant under
permutations of the tensor factors, and give an explicit
construction below. Substituting
$H=\overline\Pi\Delta\overline\Pi$ into this identity and
applying~\cref{eq:relate_derivative} gives
\begin{equation*}
 \tr(G_r\Delta^{\otimes r})=\tr(K_r\Delta^{\otimes r}),
 \qquad
 K_r:=\frac{1}{n^{\downarrow r}}
 \overline\Pi^{\otimes r}\Theta_p^{(r)}\overline\Pi^{\otimes r}.
\end{equation*}

This identity does not imply $G_r=K_r$, since $\Delta$ is
traceless and therefore does not detect terms containing an
identity factor. However, $G_r$ is centered in every copy. Since
the identity holds for every density matrix $\tau$, it identifies
$G_r$ as the operator obtained by centering $K_r$ in each copy.
More precisely, writing $\mathcal C_\rho(A)=A-\tr(\rho A)I$, we
have $G_r=\mathcal C_\rho^{\otimes r}(K_r)$
(see~\cref{lem:quantum-efron-stein}). Centering is an orthogonal
projection with respect to the $\rho$-weighted inner product, so
$\tr(\rho^{\otimes r}G_r^2)\leq\tr(\rho^{\otimes r}K_r^2)$.
Substituting into~\cref{eq:qES} and using
$n^{\downarrow r}=r!\binom nr$, we obtain
\begin{equation}
\label{eq:overview-tensor-variance}
 \begin{aligned}
 \Var(\widehat{\bm F}_p)
 &\leq\sum_{r=1}^K\frac1{r!\,n^{\downarrow r}}
   \tr\!\left[\rho^{\otimes r}
     (\overline\Pi^{\otimes r}\Theta_p^{(r)}\overline\Pi^{\otimes r})^2\right] \leq\sum_{r=1}^K\frac1{r!\,n^{\downarrow r}}
   \tr\!\left(\sigma^{\otimes r}(\Theta_p^{(r)})^2\right),
 \end{aligned}
\end{equation}
The second inequality uses $\overline\Pi\rho\overline\Pi=\sigma$
and the fact that $(PAP)^2\preceq PA^2P$ for any orthogonal
projector $P$ and Hermitian operator $A$. Thus we can bound the
variance using the derivative operators $\Theta_p^{(r)}$ without
explicitly computing the centered components $G_r$.

\paragraph{Constructing the derivative operators.}
Next, we describe how to construct the derivative operator $\Theta_p^{(r)}$, using $p(x) = x^3$ as an example. Using the cyclicity of trace gives
\begin{equation*}
 g_H(t)=g_H(0)+3t\tr(\sigma^2H)+3t^2\tr(\sigma H^2)+t^3\tr(H^3).
\end{equation*}
We want to express each derivative as the trace of an operator
against $H^{\otimes r}$.
The linear term immediately gives $\Theta_p^{(1)}=3\sigma^2$.
For the quadratic term, let $\SWAP$ exchange two tensor factors.
The identity $\tr\!\left[\SWAP(A\otimes B)\right]=\tr(AB)$ gives
\begin{equation*}
 \tr\!\left[(\sigma\otimes I+I\otimes\sigma)\SWAP H^{\otimes2}\right]
 =\tr(\sigma H^2)+\tr(H\sigma H)=2\tr(\sigma H^2).
\end{equation*}
Since $g_H''(0)=6\tr(\sigma H^2)$, we can take
$\Theta_p^{(2)}=3(\sigma\otimes I+I\otimes\sigma)\SWAP$.
For the cubic term, let $R$ cyclically permute three tensor
factors, sending $u\otimes v\otimes w$ to $v\otimes w\otimes u$.
The analogous trace identity is
$\tr\!\left[R(A\otimes B\otimes C)\right]=\tr(ABC)$.
Both $R$ and $R^\dagger$ therefore satisfy
$\tr(RH^{\otimes3})=\tr(R^\dagger H^{\otimes3})=\tr(H^3)$.
Thus $\Theta_p^{(3)}=3(R+R^\dagger)$ is Hermitian and satisfies
$\tr(\Theta_p^{(3)}H^{\otimes3})=g_H'''(0)$.

The same argument applies to a general polynomial
$p(x)=\sum_{m=0}^Kp_mx^m$. We first expand each monomial as
\begin{equation*}
 (\sigma+tH)^m
 =\sigma^m+\sum_{r=1}^m t^r
   \sum_{\substack{a_0+\cdots+a_r=m-r\\a_0,\ldots,a_r\geq0}}
   \sigma^{a_0}H\sigma^{a_1}H\cdots H\sigma^{a_r}.
\end{equation*}
A term contributes to the coefficient of $t^r$
precisely when it contains $r$ copies of $H$.
The indices $a_0,\ldots,a_r$ count the copies of $\sigma$
before, between, and after them. We keep track of their order
because $\sigma$ and $H$ need not commute.
Taking traces, summing with coefficients $p_m$, and comparing
with the Taylor expansion, we have
\begin{equation*}
 \frac{g_H^{(r)}(0)}{r!}
 =\sum_{m=r}^Kp_m
   \sum_{\substack{a_0+\cdots+a_r=m-r\\a_0,\ldots,a_r\geq0}}
   \tr\!\left(\sigma^{a_0}H\sigma^{a_1}H\cdots H\sigma^{a_r}\right).
\end{equation*}

It remains to write each term as the trace of an operator against
$H^{\otimes r}$. By cyclicity, we can combine the first and last
powers of $\sigma$, leaving a product of $r$ factors of the form
$\sigma^bH$. Letting $R_r$ be the cyclic shift that sends the
first tensor factor to the last, the trace identity gives
\begin{equation*}
 \tr(\sigma^{b_1}H\cdots\sigma^{b_r}H)
 =\tr\!\left[
 R_r(\sigma^{b_1}\otimes\cdots\otimes\sigma^{b_r})H^{\otimes r}
 \right].
\end{equation*}
The operator multiplying $H^{\otimes r}$ depends only on $\sigma$
and the exponents. Summing these operators with coefficients $p_m$
and multiplying by $r!$ therefore gives $\Theta_p^{(r)}$.

\paragraph{From matrix derivatives to scalar derivatives.}
It remains to bound each weighted second moment
$\tr(\sigma^{\otimes r}(\Theta_p^{(r)})^2)$ in~\cref{eq:overview-tensor-variance}.
We work in an eigenbasis of $\sigma$, with eigenvalues
$\alpha_1,\ldots,\alpha_d$. In this basis, we can express the
coefficients of $\Theta_p^{(r)}$ in terms of averages of the scalar derivative
$p^{(r)}$, and thereby bound the weighted second moment.
For $r=1$, we have $\Theta_p^{(1)}=p'(\sigma)$,
so
\begin{equation*}
 \tr\!\left(\sigma(\Theta_p^{(1)})^2\right)
 =\sum_a\alpha_a p'(\alpha_a)^2
 \leq\|p'\|_{\infty,[0,L]}^2,
\end{equation*}
where we used $\sum_a\alpha_a\leq1$. Similarly, for $r=2$, the matrix derivative formula gives
\begin{equation}
\label{eqn:overview-second-moment}
 \tr(\Theta_p^{(2)}H^{\otimes2})
 =g_H''(0)
 =\sum_{a,b}w_{ab}H_{ab}H_{ba}, \qquad
 w_{ab}=\int_0^1p''((1-t)\alpha_a+t\alpha_b)\,dt,
\end{equation}
and $H_{ab}$ are the entries of $H$ in the eigenbasis of
$\sigma$. To see this, first take $p(x)=x^m$ with $m\geq2$ and compute the second derivative:
\begin{equation*}
 \begin{aligned}
 g_H''(0)
 &=m\sum_{j=0}^{m-2}
   \tr\!\left(\sigma^jH\sigma^{m-2-j}H\right)=m\sum_{a,b}\left(\sum_{j=0}^{m-2}
   \alpha_a^j\alpha_b^{m-2-j}\right)H_{ab}H_{ba},
 \end{aligned}
\end{equation*}
where the second equality expands the trace in the eigenbasis
of $\sigma$. Using $y^n - x^n = (y-x)\sum_{j=0}^{n-1} x^{j} y^{n-1-j}$, we have $\sum_{j=0}^{m-2}\alpha_a^j\alpha_b^{m-2-j}
 =\frac{\alpha_b^{m-1}-\alpha_a^{m-1}}{\alpha_b-\alpha_a}.$
By linearity in $p$, the coefficient for a general polynomial is
therefore
\begin{equation*}
 w_{ab}
 =\frac{p'(\alpha_b)-p'(\alpha_a)}{\alpha_b-\alpha_a}
 =\int_0^1p''((1-t)\alpha_a+t\alpha_b)\,dt,
\end{equation*}
where the last equality is the fundamental theorem of calculus. 

To identify the entries of $\Theta_p^{(2)}$, observe that its rows and
columns are indexed by pairs. Since
$(H\otimes H)_{(c,d),(a,b)}=H_{ca}H_{db}$, expanding the trace gives
\begin{equation*}
 \tr(\Theta_p^{(2)}H^{\otimes2})
 =\sum_{a,b,c,d}(\Theta_p^{(2)})_{(a,b),(c,d)}H_{ca}H_{db}.
\end{equation*}
Taking $(c,d)=(b,a)$ gives the product $H_{ba}H_{ab}$, so
comparing with~\cref{eqn:overview-second-moment}, we can take
$(\Theta_p^{(2)})_{(a,b),(b,a)}=w_{ab}$, with all other entries zero.
Since $\Theta_p^{(2)}$ is Hermitian, its weighted second moment is a
weighted sum of squared entries:
\begin{equation*}
 \tr\!\left(\sigma^{\otimes2}(\Theta_p^{(2)})^2\right)
 =\sum_{a,b}\alpha_a\alpha_b
   \sum_{c,d}|(\Theta_p^{(2)})_{(a,b),(c,d)}|^2 =\sum_{a,b}\alpha_a\alpha_b w_{ab}^2
 \leq\|p''\|_{\infty,[0,L]}^2,
\end{equation*}
where the inequality uses $|w_{ab}|\leq\|p''\|_{\infty,[0,L]}$
and $\sum_{a,b}\alpha_a\alpha_b\leq1$.
For general $r$, the same approach gives the following bound
(see~\cref{lem:derivative-norms}):
\begin{equation}
\label{eq:overview-derivative-second-moment}
 \tr\!\left(\sigma^{\otimes r}(\Theta_p^{(r)})^2\right)
 \leq\frac{\prod_{j=0}^{r-1}(1+jL)}{(r-1)!}
       \|p^{(r)}\|_{\infty,[0,L]}^2.
\end{equation}
The matrix derivative formula in~\cref{lem:trace-derivatives}
writes $\Theta_p^{(r)}$ as a diagonal coefficient matrix times an average
of cyclic permutations. As in the quadratic case, the coefficients
are averages of $p^{(r)}$ at convex combinations of eigenvalues,
so they are bounded by $\|p^{(r)}\|_{\infty,[0,L]}$.
Bounding these coefficients and squaring leaves weighted traces
of permutation operators. Within each cycle, the matrix indices
must agree, so a cycle of length $\ell$ contributes
$\tr(\sigma^\ell)\leq L^{\ell-1}$.
The factorial denominator comes from the average over cyclic
orderings, while the product $\prod_{j=0}^{r-1}(1+jL)$ bounds
the total contribution of the resulting permutations.
Substituting~\cref{eq:overview-derivative-second-moment}
into~\cref{eq:overview-tensor-variance}, we obtain
\begin{equation}
\label{eq:overview-polynomial-variance}
 \Var(\widehat{\bm F}_p)
 \leq\sum_{r=1}^K
 \frac{\prod_{j=0}^{r-1}(1+jL)}{\binom nr(r-1)!}
 \left(\frac{\|p^{(r)}\|_{\infty,[0,L]}}{r!}\right)^2.
\end{equation}
Thus the variance is controlled by the scalar derivatives of
$p$ on $[0,L]$. Each derivative combines the monomial
contributions before we square, retaining the cancellations
between their estimation errors.

\subsection{Choosing the polynomials}

We have now reduced the problem to choosing polynomials that
approximate Lipschitz functions while controlling their derivatives
on the interval $[0,L]$. The approximation guarantees determine
how large the degree must be, while
\cref{eq:overview-polynomial-variance} bounds the resulting
variance. Suppressing constant factors, in both constructions we use
$n= d/(L\varepsilon^2)$ copies, matching the cost of estimating
the large eigenvalues. Thus we would like to take $L$ as large as
possible while keeping both the approximation and estimation errors
bounded by $O(\varepsilon)$.

\paragraph{The full spectral interval.}
A natural first choice is the basis we used for reconstruction,
$\phi_k(x)=T_k(2x/L-1)-(-1)^k$. Mapping zero to an endpoint of
$[-1,1]$ gives the finer approximation near zero, and hence total
approximation error $O(\sqrt{dL}/K)$. However, the Chebyshev
polynomials also vary most rapidly near the endpoints:
$|T_k'(\pm1)|=k^2$. More generally, we have
(see~\cref{lem:cheb-derivatives})
\begin{equation*}
 \frac{\|\phi_k^{(r)}\|_{\infty,[0,L]}}{r!}
 \leq\frac{(4k^2/L)^r}{(2r)!},\qquad 1\leq r\leq k.
\end{equation*}
Substituting into~\cref{eq:overview-polynomial-variance} and summing
gives $\Var(\widehat{\bm F}_{\phi_k})
\lesssim\exp(O(k/(nL^2)^{1/4}))$. To make the approximation error $O(\varepsilon)$, we take
$L=\varepsilon^2K^2/d$. With
$n= d/(L\varepsilon^2)$, the variance bound then grows as
$\exp(O(\sqrt K))$, allowing us to take
$K=(\log d)^2$. Each variance is then bounded by
$O(d^{1/4})$. Since
$\sum_{k=1}^Kk^{-2}=O(1)$, the estimation term
in~\cref{eq:overview-recovery} is at most
$O(Ld^{1/4})=O(\varepsilon)$. Substituting the resulting spectral bound
$L=\varepsilon^2(\log d)^4/d$ gives copy complexity
$O(d^2/(\varepsilon^4(\log d)^4))$.

\paragraph{Moving to the interior.}
To improve the dependence on $\varepsilon$, we would like to
increase the degree as the desired accuracy increases. The large
endpoint derivatives are an obstacle, which suggests placing the
spectral interval strictly inside the Chebyshev interval. We use
$\psi_k(x)=T_k(x/L-1/2)-T_k(-1/2)$, so that the affine rescaling
maps $[0,L]$ to $[-1/2,1/2]$ and $\psi_k$ again vanishes at zero.
Away from the endpoints,
the derivatives are smaller (see~\cref{lem:cheb-derivatives}):
\begin{equation*}
 \frac{\|\psi_k^{(r)}\|_{\infty,[0,L]}}{r!}
 \leq\frac{(2k/L)^r}{r!},\qquad 1\leq r\leq k.
\end{equation*}
To obtain an approximation in this basis, we extend each Lipschitz
test function constantly outside $[0,L]$ to $[-L/2,3L/2]$ and apply
the same Jackson construction on this larger interval. Zero is now
in the interior, so we use the uniform approximation error
$O(L/K)$, which sums to $O(dL/K)$ over the spectrum. The constant
extension leaves the integral of the squared derivative unchanged,
so the coefficient bound in~\cref{eq:overview-derivative-energy}
still holds. Repeating the reconstruction argument with $\psi_k$
in place of $\phi_k$ therefore gives
\begin{equation*}
 \E\TV(\widehat{\bm z}_{\mathrm{int}},\alpha)
 \lesssim\frac{dL}{K}
 +L\sum_{k=1}^K\frac{\Var(\widehat{\bm F}_{\psi_k})}{k^2}.
\end{equation*}

The approximation error is larger, but the smaller derivatives
allow us to compensate by increasing the degree. We now need
$K= dL/\varepsilon$ to make the approximation error
$O(\varepsilon)$. With the same copy budget
$n= d/(L\varepsilon^2)$, substituting the derivative bounds
into~\cref{eq:overview-polynomial-variance} allows us to take
$L=(\log d)^2/d$, and hence
$K=(\log d)^2/\varepsilon$. For these choices, we obtain
$\Var(\widehat{\bm F}_{\psi_k})/k^2\lesssim d^{1/4}/(nL^2)$.
The estimation term is therefore at most
$O(Kd^{1/4}/(nL))=O(\varepsilon)$ as well. The threshold $B=L/1.1$ is
now independent of $\varepsilon$, giving copy complexity
$O(d^2/(\varepsilon^2(\log d)^2))$.

\section{Discussion and outlook}
\label{sec:discussion}

\paragraph{Applications to learning other symmetric properties.}
A natural question is whether our local Chebyshev moment matching
methods can determine the optimal copy complexity of learning or
testing other symmetric properties, such as von Neumann entropy
estimation~\cite{AISW20,GW26} and rank testing~\cite{OW21}.
Known upper and lower bounds for these problems agree up to
polylogarithmic factors~\cite{wang2026nearlytightlowerbounds},
but the sharp rates remain unresolved.
Our variance bounds apply to general polynomial statistics,
suggesting that approximations tailored to the property of
interest may help close these gaps.
More broadly, can our techniques improve the copy complexity of
estimating trace distance and fidelity between two unknown quantum
states~\cite{GP22,UNWT25,LT26}? These quantities also depend on the relative
eigenbases, so extending our framework would require estimating
information beyond the two spectra separately.

\paragraph{Optimal two-bucket algorithms for learning sorted distributions?}
The classical analogue of spectrum estimation is to estimate an unknown distribution on $d$ elements up to permutation of its labels, with TV distance error $\epsilon$. 
Valiant and Valiant~\cite{VV11a,VV17} gave the first algorithm with sample complexity $O(d/\log d)$ for constant $\epsilon$, beating the $\Theta(d)$ samples required to learn the labeled distribution.  
Their approach separates the distribution into two buckets, using empirical estimates for the large probabilities and a
linear program to recover the small probabilities.

Han, Jiao, and Weissman~\cite{HJW18} subsequently improved their result using a \emph{multi-bucket} strategy.
For any fixed $\gamma\in(0,1)$, their algorithm achieves the optimal sample complexity $\Theta(d/(\varepsilon^2\log d))$ when $\varepsilon\geq 1/d^{1-\gamma}$.
At the other extreme, when $\varepsilon\leq 1/d$, the optimal sample complexity is $\Theta(d/\varepsilon^2)$, which is achieved by the empirical distribution directly.
Between these two regimes, the optimal sample complexity remains unsettled. 

This raises a natural question: As we solve quantum spectrum estimation optimally using only two buckets, 
can our local Chebyshev moment matching method yield a two-bucket algorithm with optimal sample complexity across all accuracy regimes?

\paragraph{Is adaptivity necessary for spectrum estimation?}
Both our entangled and unentangled algorithms first learn a small-bucket projector and then use it to choose measurements on fresh copies.
However, for entangled measurements, adaptivity is unnecessary. 
It is well-known that weak Schur sampling is the optimal measurement for learning any spectral property of a quantum state. 
So there exists classical postprocessing such that spectrum estimation can be done in a sample-optimal way. 
What is unclear is how to make this postprocessing and its analysis explicit. 
Such a procedure would achieve the optimal copy bound without learning a projector or choosing measurements adaptively.
For unentangled measurements, we ask the same question: can the copy bound in \cref{thm:unentangled} be achieved when all single-copy measurements are fixed in advance?

\paragraph{Optimal lower bound for unentangled spectrum estimation.}
We conjecture that the sample complexity in \Cref{thm:unentangled} is optimal among all adaptive single-copy protocols.  
The current best unentangled lower bound is inherited from entangled spectrum estimation. Therefore it is an outstanding open question to improve the unentangled lower bound.

\section{Preliminaries}

We write random variables in bold.  
All polynomials in this paper have real coefficients. 
For a Hermitian matrix $A$, let $\spec(A)$ denote its eigenvalue vector, listed in nonincreasing order and with zero eigenvalues retained.    
We write $a^{\uparrow j}=a(a+1)\cdots(a+j-1)$ for the rising factorial and $n^{\downarrow m}=n(n-1)\cdots(n-m+1)$ for the falling factorial.
We use $C,c>0$ for universal constants that may change from line to line. 
We write $\mathrm{D}(\calH)$ for the set of density operators on a Hilbert space $\calH$. 

\begin{definition}[Quantum measurement models]
A quantum measurement on a Hilbert space $\mathcal H$ is described by a positive operator-valued measure (POVM): a collection $\{M_i\}_i$ of PSD operators on $\mathcal H$ satisfying $\sum_i M_i=I$. When applied to a quantum state $\rho$, this measurement returns outcome $i$ with probability $\tr(M_i\rho)$. We also allow continuous outcomes, with sums replaced by integrals and probabilities by probability densities.
Given $n$ copies of a quantum state $\rho\in \mathrm{D}(\C^{d})$, we consider two measurement models:
\begin{itemize}
  \item In the \emph{entangled measurement model}, we allow an arbitrary POVM on $(\C^d)^{\otimes n}$. So each $M_i \in \C^{d^n \times d^n}$ and outcome $i$ occurs with probability $\tr(M_i\rho^{\otimes n})$.
  \item In the \emph{unentangled measurement model}, each copy of $\rho$ is measured separately using a POVM on $\C^d$. The choice of POVM may depend on previous measurement outcomes, but only classical information is retained between copies.
\end{itemize}
\end{definition}

\begin{definition}[Unnormalized Wasserstein-$1$ distance]
Let $p,q\in\mathbb R^d$ have entries sorted in the same order,
and let $\mu_p:=\sum_{i=1}^d\delta_{p_i},$ and $\mu_q:=\sum_{i=1}^d\delta_{q_i}$ be their counting measures, where $\delta_x$ is the unit
point mass at $x$.
Their unnormalized Wasserstein-$1$ distance is $W_1(\mu_p,\mu_q)=\sum_{i=1}^d|p_i-q_i|.$
Thus, writing $\TV(p,q):=\frac12\|p-q\|_1$, we have $W_1(\mu_p,\mu_q)=2\TV(p,q).$
\end{definition}

\begin{definition}[Kantorovich-Rubinstein duality of Wasserstein-$1$ distance]
For the counting measures $\mu_p,\mu_q$ above,
\[
W_1(\mu_p,\mu_q)
=\sup_{\operatorname{Lip}(f)\leq1}
  \left|\sum_{i=1}^d f(p_i)-\sum_{i=1}^d f(q_i)\right|,
\]
where the supremum is over all $1$-Lipschitz functions
$f:\mathbb R\to\mathbb R$.
\end{definition}

Our results use Chebyshev polynomials for approximating Lipschitz functions on a bounded domain. We give the definition of the Chebyshev polynomials. 

\begin{definition}[Chebyshev Polynomials]\label{def:chebyshev-poly}
    For $k \in \N$, the $k$-th Chebyshev polynomial is a degree $k$ polynomial denoted by $T_k$. It is recursively defined as follows: 
    $$T_0(x) = 1, \qquad T_1(x) = x, \qquad  \text{and} \qquad T_{k}(x) = 2x T_{k-1}(x) - T_{k-2}(x) .$$
    On the interval $[-1,1]$, the Chebyshev polynomial can alternatively be defined via the trigonometric definition $T_k(\cos \theta) = \cos(k \theta)$.
\end{definition}

We also note down the orthogonality property of the Chebyshev polynomials, which we will exploit in our work. A proof can be found in \cite[Section 1.5]{rivlin1990chebyshev}.

\begin{fact}[Orthogonality of Chebyshev Polynomials]
    The Chebyshev polynomials are orthogonal with respect to the weight function $w(x) = 1/\sqrt{1-x^2}$. In particular, for $j \neq k \in \N$, $\int_{-1}^{1} T_k(x) T_j(x) w(x) \diff x = 0$.
\end{fact}

We also note down Jackson's approximation theorem, which says that the \emph{damped} Chebyshev series of a Lipschitz function is a \emph{good} polynomial approximation. A modern proof can be found in \cite[Fact 3.2]{BKM22}. We will use a more refined version of this in \Cref{lem:jackson-endpoint-improvement}. 
\begin{fact}[Jackson's Theorem \cite{Jackson:1930}] \label{fact:jackson}
  Let $f : [-1,1] \to \R$ be an $\ell$-Lipschitz function, and let its Chebyshev series be $f(t) = \gamma_0 + \sum_{k\geq 1} \gamma_k T_k(t)$. Then, for any $K \in \N$, there exist damping factors $\eta_1,\dots,\eta_K \in [0,1]$ depending only on $K$ and the coefficient index such that the polynomial
  $$f_K := \gamma_0 + \sum_{k =1}^K \eta_k \gamma_k T_k(t) $$
  satisfies $\norm{f-f_K}_{\infty, [-1,1]} \leq \frac{\ell}{K}$.
\end{fact}

\section{Bucketing algorithms}
\label{sec:reduction}

The first stage of both our spectrum estimation algorithms is to apply a \emph{bucketing algorithm}. 
Given a threshold $B$ and an accuracy $\eta$, the bucketing algorithm learns the eigenvalues greater than $B$ to error $\eta$ and returns a small-bucket projector $\overline{\Pi}$.
The second stage is to learn the small eigenvalues of the small-bucket state $\overline{\Pi} \rho\overline{\Pi}$. 

This two-stage bucketing framework was adopted for both unentangled spectrum estimation in \cite{PTTW26} and entangled spectrum estimation in \cite{PSTW26}. We use the same bucketing algorithms and we differ in the second stage of learning the small eigenvalues, which is the reason for our improved sample complexities. We will describe the second stage in detail in \Cref{sec:geometries} and give a full description of our spectrum estimation algorithms in \Cref{sec:rates}. In this section, we import the guarantees of the bucketing algorithms derived in \cite{PTTW26,PSTW26}.

\begin{definition}[Bucketing algorithm]\label{def:bucketing-algorithm}
Fix a threshold $0<B<1$ and an accuracy $0<\eta<1$.  A \emph{bucketing
algorithm} at $(B,\eta)$ of cost $N$ is a procedure which, using $N$ copies of $\rho$, returns an orthogonal projector $\bm{\Pi}$, the small-bucket projector
$\overline{\bm{\Pi}}=I-\bm{\Pi}$, and a sorted list of nonnegative numbers $\widehat{\bm{\alpha}}_{\mathrm{Large}}$ of length $\rank(\bPi)$ such that the following is true. 
Put
\[
 \bm{\sigma}=\overline{\bm{\Pi}}\rho\overline{\bm{\Pi}},
 \qquad
 \bm{\alpha}_{\mathrm{Large}}
 =\spec\bigl(\bm{\Pi}\rho\bm{\Pi}|_{\operatorname{ran}\bm{\Pi}}\bigr).
\]
We call $\operatorname{ran}\overline{\bm{\Pi}}$ the \emph{small-bucket subspace}
and $\bm{\sigma}$ the \emph{small-bucket state}.
The state $\bm{\sigma}$ always satisfies $\tr(\bm{\sigma})\leq1$.
With probability at least $0.99$, all three bounds below hold:
\begin{enumerate}[label=(\roman*),ref=(\roman*)]
\item\label{item:pack-list} \emph{Low error in learning the large eigenvalues.}
The large eigenvalues of $\rho$ can be estimated to error $\eta$: 
\[
 \TV(\bm{\alpha}_{\mathrm{Large}}, \, \widehat{\bm{\alpha}}_{\mathrm{Large}})\leq\eta.
\]
\item\label{item:pack-small} \emph{Low misclassification error.}
The small eigenvalues of $\rho$ are classified into the small bucket:
\[
 0\preceq\bm{\sigma}\preceq1.1B\,I_d.
\]
\item\label{item:pack-pinch} \emph{Low alignment error.}
The full spectrum of $\rho$ is disturbed by at most error $\eta$: 
\[
 \TV\!\left(\spec(\rho), \, \spec(\bm{\Pi}\rho\bm{\Pi}+\bm{\sigma})\right)\leq\eta.
\]
\end{enumerate}
\end{definition}

We use the entangled bucketing algorithm which follows from
\cite[Definition~5.2, Lemma~5.3, and Lemma~5.4(2)]{PSTW26}.

\begin{proposition}[Entangled bucketing algorithm]
\label{prop:bucketing-entangled}
For every $0<B<1$ and $0<\eta<1$, there is a bucketing algorithm at
$(B,\eta)$ of cost
\[
 N=O\!\left(\frac{d}{B\eta^2}\right),
\]
implemented using entangled measurements.
\end{proposition}

We use the unentangled bucketing algorithm which follows from
\cite[Definition~6.1, Theorem~6.2, and Equation~(31) in its proof]{PTTW26}.

\begin{proposition}[Unentangled bucketing algorithm]
\label{prop:bucketing-single}
For every $0<B<1$ and $0<\eta\leq1/10$, there is a bucketing algorithm at
$(B,\eta)$ of cost
\[
 N=O\!\left(\frac{d}{B^2\eta^2}\right),
\]
implemented using unentangled measurements.
\end{proposition}

The non-zero eigenvalues of $\bm{\Pi}\rho\bm{\Pi}+\bm{\sigma}$ are the non-zero entries of $\bm{\alpha}_{\mathrm{Large}}$ together with the non-zero eigenvalues of $\bm{\sigma}$.  
By \Cref{item:pack-pinch}, it is close to the desired spectrum of $\rho$. 
Therefore, it remains to estimate the spectrum of the small-bucket state $\bm{\sigma}$. 

In \Cref{sec:geometries}, we give a self-contained analysis and two methods for recovering the spectrum of a subnormalized state $0\preceq \sigma\preceq L I_d$ using estimates of $\tr(p(\sigma))$ for some polynomial $p$. 
To apply these two methods to $\bm{\sigma} = \overline\bPi \rho \overline\bPi$, we need to obtain estimates of $\tr(p(\bsigma))$ using copies of $\rho$ and the small-bucket projector $\overline\bPi$ from the bucketing algorithm.
The corresponding entangled and unentangled measurements are described in \Cref{sec:small-bucket-measurements}, and we bound the variance of these estimators in the rest of \Cref{sec:backends}.

\section{Spectrum recovery from Chebyshev moments}
\label{sec:geometries}

This section is entirely classical and applies to both measurement
models. Given unbiased estimates of the monomial moments, we
construct a spectrum by matching estimated Chebyshev moments.
We consider two affine rescalings and prove guarantees for
recovery on the full Chebyshev interval
(\Cref{cor:cheb-recovery}) and in its interior
(\Cref{cor:buffered-recovery}). Both guarantees bound the expected
total-variation error by an approximation term plus a weighted
sum of moment variances. The subsequent variance analysis in
\Cref{sec:backends} will bound these variances for each
measurement model.

\paragraph{Setup.}
Fix a dimension $d\geq1$, a number $L>0$, and an integer $K\geq1$.
Let $\alpha=(\alpha_1,\ldots,\alpha_d)$ be the true spectrum, and
assume that $\alpha$ belongs to the compact set
\[
 \calZ_d=\braces[\bigg]{z\in[0,L]^d:z_1\geq\cdots\geq z_d,
                         \ \sum_{i=1}^d z_i\leq1 }.
\]
Both estimators will choose a spectrum from $\calZ_d$.

Let $\widehat{\bm M}_1,\ldots,\widehat{\bm M}_K$ be unbiased real
estimates of the monomial moments $\sum_i\alpha_i^m$, with finite
variances, and set $\widehat{\bm M}_0=d$.  Recall that we write random variables in bold. 
For a polynomial
$p(x)=\sum_{m=0}^Kp_mx^m$, define
\[
 \widehat{\bm F}_p=\sum_{m=0}^Kp_m\widehat{\bm M}_m.
\]
Unbiasedness gives $\E\widehat{\bm F}_p=\sum_i p(\alpha_i)$, and
therefore
\begin{equation}\label{eq:polynomial-estimation-variance}
 \Var(\widehat{\bm F}_p)
 =\E\!\left[
   \left(\widehat{\bm F}_p-\sum_{i=1}^d p(\alpha_i)\right)^2
       \right].
\end{equation}
No independence assumption on the moment estimates is needed.

For $z\in\calZ_d$, let $\mu_z=\sum_{i=1}^d\delta_{z_i}$ be its
counting measure, retaining all zero entries. Thus every counting
measure has mass $d$, while the constraint $\sum_i z_i\leq1$
bounds its first moment. For sorted spectra $\alpha,z\in\calZ_d$, recall from \Cref{sec:overview-cheb-matching} that 
\begin{equation}\label{eq:w1-list}
 W_1(\mu_\alpha,\mu_z)
 =\sum_{i=1}^d|\alpha_i-z_i|
 =2\TV(\alpha,z).
\end{equation}

Throughout, $C,c>0$ denote universal constants
that may change from line to line.

Let $T_k$ denote the degree-$k$ Chebyshev polynomial of the first kind. 
We define two rescalings that map the spectrum interval $[0,L]$ to $[-1,1]$ and $[-1/2,1/2]$, respectively.
The following two propositions define the estimators and bound their recovery errors.

\begin{proposition}[Recovery on the full Chebyshev interval]
\label{cor:cheb-recovery}
For $1\leq k\leq K$, define
\begin{equation}\label{eq:two-chebyshev-phi}
 \phi_k(x):=T_k\!\left(\frac{x-L/2}{L/2}\right)-(-1)^k.
\end{equation}
Note that $\phi_k(0)=0$.
Define the spectrum estimator as
\begin{equation*}
 \widehat{\bm z}_{\rm full}\in\argmin_{z\in\calZ_d}
 \sum_{k=1}^K\frac1{k^2}
 \left(\widehat{\bm F}_{\phi_k}-\sum_{i=1}^d\phi_k(z_i)\right)^2.
\end{equation*}
Then
\begin{equation}\label{eq:cheb-recovery}
 \E\TV(\alpha,\widehat{\bm z}_{\rm full})
 \leq C\frac{\sqrt{dL}}K
       +CL\sum_{k=1}^K
          \frac{\Var(\widehat{\bm F}_{\phi_k})}{k^2}.
\end{equation}
\end{proposition}

\begin{proposition}[Recovery in the interior of the Chebyshev interval]
\label{cor:buffered-recovery}
For $1\leq k\leq K$, define
\begin{equation}\label{eq:two-chebyshev-psi}
 \psi_k(x):=T_k\!\left(\frac{x-L/2}{L}\right)-T_k(-1/2).
\end{equation}
Note that $\psi_k(0)=0$.
Define the spectrum estimator as
\begin{equation*}
 \widehat{\bm z}_{\rm int}\in\argmin_{z\in\calZ_d}
 \sum_{k=1}^K\frac1{k^2}
 \left(\widehat{\bm F}_{\psi_k}-\sum_{i=1}^d\psi_k(z_i)\right)^2.
\end{equation*}
Then
\begin{equation}\label{eq:buffered-recovery}
 \E\TV(\alpha,\widehat{\bm z}_{\rm int})
 \leq C\frac{dL}{K}
       +CL\sum_{k=1}^K
          \frac{\Var(\widehat{\bm F}_{\psi_k})}{k^2}.
\end{equation}
\end{proposition}

Matching Chebyshev moments on the full interval $[-1,1]$ gives an approximation error bound of $O(\sqrt{dL}/K)$, compared with $O(dL/K)$ for the interior $[-1/2,1/2]$. 
This improvement comes from the more accurate polynomial approximation near the Chebyshev endpoints.

However, the benefit of the interior method is that the Chebyshev derivatives are smaller in the interior, which will help control the variances of the estimated moments. 
For example, 
\begin{equation*}
 \|\phi_k'\|_{\infty,[0,L]}=\frac{2k^2}{L},
 \qquad
 \|\psi_k'\|_{\infty,[0,L]}\leq\frac{2k}{L}.
\end{equation*}
The variance analysis in \Cref{sec:backends} uses these derivatives and their higher-order counterparts to bound the moment variances.
We record the higher-order bounds and their proof in \Cref{app:chebyshev-derivatives}.

\paragraph{Proof layout.} We give the proof idea of \Cref{cor:cheb-recovery} and \Cref{cor:buffered-recovery} in \Cref{sec:moment-matching-motivation}. In \Cref{sec:approximation-error} we bound the approximation error, which corresponds to the first terms in \cref{eq:cheb-recovery} and \cref{eq:buffered-recovery}. In \Cref{sec:moment-error}, we justify the second terms in \cref{eq:cheb-recovery} and \cref{eq:buffered-recovery}. Then, finally, we give the proof of the two propositions.

\subsection{Proof idea}
\label{sec:moment-matching-motivation}

Our estimator selects a candidate spectrum $z\in \calZ_d$ whose moments best match the estimated moments.  Recall from \cref{eq:overview-wasserstein-duality} that 
\begin{equation}\label{eq:wasserstein-dual}
 W_1(\mu_\alpha, \mu_z) = 2\TV(\alpha,z)
 =\sup_{\operatorname{Lip}(f)\leq1}
   \left|\sum_{i=1}^d f(\alpha_i)-\sum_{i=1}^d f(z_i)\right|.
\end{equation}
Adding a constant to $f$ does not change the
difference so we can assume $f(0)=0$. As in \Cref{sec:overview-cheb-matching}, a natural approach
is to approximate $f$ by a low degree polynomial
$P_f(x)=\sum_{m=0}^K b_mx^m$. Adding and subtracting the sums of
$P_f$ over the two spectra gives
\begin{equation}\label{eq:approximation-moment-decomposition}
 \left|\sum_i f(\alpha_i)-\sum_i f(z_i)\right|
 \leq
 \underbrace{\sum_i\bigl(|f(\alpha_i)-P_f(\alpha_i)|
                   +|f(z_i)-P_f(z_i)|\bigr)}_{\text{approximation error}}
 +\underbrace{\left|\sum_{m=1}^K b_m
   \left(\sum_i\alpha_i^m-\sum_i z_i^m\right)\right|}
   _{\text{moment discrepancy}}.
\end{equation}

The first term depends on how accurately $P_f$ approximates
$f$ at the two spectra. The second depends on how closely
their moments agree, weighted by the coefficients of $P_f$.
Large coefficients can amplify the moment differences, so
an accurate approximation alone is not enough: we also
need control of its coefficients.

Following the Chebyshev moment-matching method of \cite{MMRS25},
we express the approximating polynomials in the Chebyshev basis. The Jackson construction below
provides both an approximation bound and a useful
description of the coefficients: each nonconstant
coefficient is obtained by damping the corresponding
Chebyshev coefficient of the function being approximated.

We first bound the approximation error for the two
rescalings. We then bound the coefficients and show how
the least-squares objective controls the moment
discrepancy through the variances of the estimated
Chebyshev moments. Finally, we combine these bounds
using a function attaining the supremum in
\cref{eq:wasserstein-dual}.

\subsection{Controlling the approximation error}\label{sec:approximation-error}

\begin{lemma}[Jackson approximation near the endpoints]
\label{lem:jackson-endpoint-improvement}
Let $f:[-1,1]\to\R$ be $1$-Lipschitz and let $K\geq1$. Write the Chebyshev expansion of $f$ as
\[
 f(t)=\gamma_0+\sum_{k\geq1}\gamma_kT_k(t).
\]
There exist damping factors $\eta_1,\ldots,\eta_K\in[0,1]$,
depending only on $K$ and the coefficient index, such that
the polynomial
\[
 Q_f(t)=\gamma_0+\sum_{k=1}^K a_kT_k(t),
 \qquad a_k=\eta_k\gamma_k,
\]
satisfies
\[
 |f(t)-Q_f(t)|
 \leq C\left(\frac{\sqrt{1-t^2}}{K}+\frac{1}{K^2}\right),
 \qquad t\in[-1,1].
\]
Moreover, if $f(-1)=0$, the anchored polynomial
$P_f(t)=Q_f(t)-Q_f(-1)  =\sum_{k=1}^K a_k\bigl(T_k(t)-(-1)^k\bigr)$ satisfies
\[
 |f(t)-P_f(t)|\leq C\frac{\sqrt{1+t}}{K},
 \qquad t\in[-1,1].
\]
Here $C$ is a universal constant.
\end{lemma}
We give a proof of this in \Cref{sec:jackson-approximation}.
Note that the anchored polynomial $P_f$ has only the constant coefficient different from $Q_f$, and the other Chebyshev coefficients are the same. The first estimate gives error $O(1/K)$ throughout the interval,
improving to $O(1/K^2)$ at its endpoints. The second estimate
uses exact agreement at $-1$ to obtain a bound that vanishes
there.

The first terms in \Cref{cor:cheb-recovery} (\cref{eq:cheb-recovery}) and \Cref{cor:buffered-recovery} (\cref{eq:buffered-recovery}) correspond to the approximation error (\cref{eq:approximation-moment-decomposition}). With Jackson's theorem in hand, we will now justify these terms. 
Note that our spectra lie in $[0,L]$, so we consider two affine maps into
the Chebyshev interval:
\[
 x\longmapsto\frac{x-L/2}{L/2}\in[-1,1],
 \qquad
 x\longmapsto\frac{x-L/2}{L}\in[-1/2,1/2], 
\]
corresponding to the rescaling of $T_k$ in the definitions of $\phi_k$ in \cref{eq:two-chebyshev-phi} and $\psi_k$ in \cref{eq:two-chebyshev-psi}.

The first map sends the spectral endpoints $0$ and $L$ to the
Chebyshev endpoints $-1$ and $1$. It therefore preserves the
endpoint improvement. In particular, rescaling the anchored
estimate gives approximation error $O(\sqrt{Lx}/K)$ at
$x\in[0,L]$.

Replacing $f$ by $P_f$ in the two spectral sums contributes
the approximation error in
\cref{eq:approximation-moment-decomposition}.
For the full-interval construction which sends $[0,L]$ to $[-1,1]$, and any $z\in\calZ_d$, we get using \Cref{lem:jackson-endpoint-improvement}
\begin{equation}\label{eq:full-approximation-error}
 \sum_{i=1}^d|f(\alpha_i)-P_f(\alpha_i)|
  +\sum_{i=1}^d|f(z_i)-P_f(z_i)|\leq\frac{C\sqrt L}{K}
   \left(\sum_{i=1}^d\sqrt{\alpha_i}
        +\sum_{i=1}^d\sqrt{z_i}\right)
 \leq\frac{C\sqrt{dL}}{K}.
\end{equation}
The last inequality uses Cauchy--Schwarz and the constraints
$\sum_i\alpha_i\leq1$ and $\sum_i z_i\leq1$.
This gives the approximation term in
\cref{eq:cheb-recovery}.

The second map sends the spectral endpoints $0$ and $L$
to $-1/2$ and $1/2$, strictly inside the Chebyshev interval.
On this smaller interval, $\sqrt{1-t^2}$ is bounded away
from zero, so the estimate in
\Cref{lem:jackson-endpoint-improvement} gives the uniform
error bound $\|f-P_f\|_{\infty,[0,L]}\leq CL/K$ after
rescaling. Thus this estimate does not give an additional
improvement at the spectral endpoints. Summing the errors
over both spectra gives, for any $z\in\calZ_d$,
\begin{equation}\label{eq:interior-approximation-error}
 \sum_{i=1}^d|f(\alpha_i)-P_f(\alpha_i)|
  +\sum_{i=1}^d|f(z_i)-P_f(z_i)|\leq2d\|f-P_f\|_{\infty,[0,L]}
 \leq\frac{CdL}{K}.
\end{equation}
This gives the approximation term in
\cref{eq:buffered-recovery}. 
We remark that more generally, we could map $[0,L]$ to $[-a,a]$ for any
$a\in(0,1)$. Taking $a$ small increases the Jackson
approximation bound to $O(L/(aK))$, while taking $a$ close
to $1$ worsens the derivative bounds near the Chebyshev
endpoints. The choice $a=1/2$ is convenient because it stays
away from both extremes. Any fixed $a\in(0,1)$ gives the
same asymptotic guarantees, with constants depending on $a$.

We now justify the moment discrepancy in
\cref{eq:approximation-moment-decomposition} to obtain the second
terms in \cref{eq:cheb-recovery,eq:buffered-recovery}.

\subsection{Controlling the moment discrepancy}\label{sec:moment-error}

\paragraph{Fitting the estimated moments.}
We first begin with justifying the minimizer objective in \Cref{cor:cheb-recovery} and \Cref{cor:buffered-recovery}. Recall that the moment discrepancy is the term $\left|\sum_i P(\alpha_i)-\sum_i P(z_i)\right|$ in \cref{eq:approximation-moment-decomposition}. Suppose $P=\sum_{k=1}^K a_k q_k$, where $q_k$ denotes either
$\phi_k$ for every $k$ or $\psi_k$ for every $k$.  If
$\Delta_k(z) :=\sum_i(q_k(\alpha_i)-q_k(z_i))$, then Cauchy--Schwarz gives us 
\begin{equation}\label{eq:weighted-coefficient-cauchy-schwarz}
 \left|\sum_i P(\alpha_i)-\sum_i P(z_i)\right|
 \leq
 \left(\sum_{k=1}^K k^2a_k^2\right)^{1/2}
 \left(\sum_{k=1}^K\frac{\Delta_k(z)^2}{k^2}\right)^{1/2}.
\end{equation}
This inequality suggests controlling
$\sum_k\Delta_k(z)^2/k^2$. However, the true moments
$\sum_i q_k(\alpha_i)$ are unknown. We therefore replace them
by their estimates and minimize the observable objective
\[
 \sum_{k=1}^K\frac1{k^2}
 \left(\widehat{\bm F}_{q_k}-\sum_i q_k(z_i)\right)^2.
\]
This replacement can be justified by comparing the estimated
spectrum $\widehat{\bm z}$ with the true spectrum $\alpha$.
Since $\alpha$ is feasible, minimality gives
\[
 \sum_{k=1}^K\frac1{k^2}
 \left(\widehat{\bm F}_{q_k}
       -\sum_i q_k(\widehat{\bm z}_i)\right)^2
 \leq
 \sum_{k=1}^K\frac1{k^2}
 \left(\widehat{\bm F}_{q_k}-\sum_i q_k(\alpha_i)\right)^2.
\]
The triangle inequality for the weighted Euclidean norm
therefore implies
\begin{equation} \label{eq:moment_expectation_error}
 \sum_{k=1}^K\frac{\Delta_k(\widehat{\bm z})^2}{k^2}
 \leq
 4\sum_{k=1}^K\frac1{k^2}
 \left(\widehat{\bm F}_{q_k}-\sum_i q_k(\alpha_i)\right)^2.
\end{equation}

Taking expectations and using unbiasedness yields
\begin{equation*} 
    \E\left[\sum_{k=1}^K
       \frac{\Delta_k(\widehat{\bm z})^2}{k^2}\right]
 \leq
 4\sum_{k=1}^K\frac{\Var(\widehat{\bm F}_{q_k})}{k^2}.
\end{equation*}

Thus minimizing the observable objective controls the weighted
moment discrepancies by the weighted sum of estimation
variances. It remains to control the coefficient factor
$\sum_{k=1}^K k^2a_k^2$ in
\cref{eq:weighted-coefficient-cauchy-schwarz}.

\paragraph{Coefficient bounds for the Jackson approximations.}
To control the
moment discrepancy through
\cref{eq:weighted-coefficient-cauchy-schwarz}, we now need a bound on
$\sum_k k^2a_k^2$. Recall that $a_k$ is the coefficient of the 
polynomial approximating $f$, in the basis $\phi_k$ or $\psi_k$. By the Jackson
construction in \Cref{lem:jackson-endpoint-improvement},
these coefficients are obtained by damping the Chebyshev
coefficients of the rescaled Lipschitz function. We use this
property to bound $\sum_{k=1}^K k^2a_k^2$ in terms of
$\int_0^L|f'(x)|^2\diff x$. Orthogonality of the Chebyshev polynomials relates this sum to the derivative of the function being approximated. We now prove a bound on the weighted coefficient.

\begin{lemma}[Coefficient energy of the Jackson approximations]
\label{lem:jackson-coefficient-energy}
Let $f:[0,L]\to\R$ be $1$-Lipschitz with $f(0)=0$.
For either choice $q_k=\phi_k$ for all $k$ or $q_k=\psi_k$
for all $k$, there exists a polynomial
\[
 P_f(x)=\sum_{k=1}^K a_kq_k(x)
\]
such that, for every $x\in[0,L]$,
\[
 |f(x)-P_f(x)|
 \leq
 \begin{cases}
  C\sqrt{Lx}/K,&q_k=\phi_k,\\
  CL/K,&q_k=\psi_k,
 \end{cases}
\]
and
\begin{equation}\label{eq:jackson-coefficient-bound}
 \sum_{k=1}^K k^2a_k^2
 \leq CL\int_0^L|f'(x)|^2\diff x.
\end{equation}
Here $f'$ denotes the almost-everywhere derivative of $f$.
\end{lemma}

\begin{proof}
For the full-interval basis $q_k=\phi_k$, define
\[
 g(t)=f\!\left(\frac L2(1+t)\right),
 \qquad t\in[-1,1].
\]
For the interior basis $q_k=\psi_k$, define
\[
 g(t)=
 \begin{cases}
  0,&-1\leq t<-1/2,\\
  f(L/2+Lt),&-1/2\leq t\leq1/2,\\
  f(L),&1/2<t\leq1.
 \end{cases}
\]
In the second case, we extend the rescaled function constantly
outside $[-1/2,1/2]$ so that it is defined on the entire
Chebyshev interval. Since $f(0)=0$, this extension is continuous.
The function $g$ is $L/2$-Lipschitz in the full-interval case
and $L$-Lipschitz in the interior case.

In either case, write the Chebyshev expansion of $g$ as
\[
 g(t)=\sum_{k\geq0}\gamma_kT_k(t).
\]

We construct $P_f$ by applying
\Cref{lem:jackson-endpoint-improvement} to a rescaled function
and subtracting the resulting polynomial's value at $x=0$.
We then use coefficient damping to prove
\cref{eq:jackson-coefficient-bound}.
Using the proof of Lemma 13 of \cite{MMRS25}, we get that 
\begin{align*}
 \sum_{k\geq1}k^2\gamma_k^2
 &=\frac2\pi\int_{-1}^1|g'(t)|^2\sqrt{1-t^2}\diff t.
\end{align*}

For either Jackson approximation, $a_k=\eta_k\gamma_k$ with
$0\leq\eta_k\leq1$.  Consequently,
\begin{equation}\label{eq:coefficient-damping-transfer}
 \sum_{k=1}^K k^2a_k^2
 \leq\sum_{k\geq1}k^2\gamma_k^2
 =\frac2\pi\int_{-1}^1|g'(t)|^2\sqrt{1-t^2}\diff t.
\end{equation}
Using $\sqrt{1-t^2}\leq1$ in
\cref{eq:coefficient-damping-transfer}, we obtain
\[
 \sum_{k=1}^K k^2a_k^2
 \leq\frac2\pi\int_{-1}^1|g'(t)|^2\diff t
 \leq\frac{2L}{\pi}\int_0^L|f'(x)|^2\diff x.
\]
For the last inequality, the substitution $x=L(1+t)/2$
in the full-interval case gives
$\int_{-1}^1|g'|^2=(L/2)\int_0^L|f'|^2$.
In the interior case, $g'$ vanishes outside $[-1/2,1/2]$,
and the substitution $x=L/2+Lt$ gives
$\int_{-1}^1|g'|^2=L\int_0^L|f'|^2$.
This proves \cref{eq:jackson-coefficient-bound}.
\end{proof}

In \Cref{sec:common-recovery-proof}, for each candidate spectrum
$z\in\calZ_d$, we will construct a $1$-Lipschitz function
$f:[0,L]\to\R$ with $f(0)=0$ such that
\[
  \sum_{i=1}^d f(\alpha_i)-\sum_{i=1}^d f(z_i)
  =W_1(\mu_\alpha,\mu_z)=2\TV(\alpha,z).
\]
Thus this function realizes the supremum in
\cref{eq:wasserstein-dual}. It also satisfies
\[
  \int_0^L |f'(x)|^2\diff x\leq 2\TV(\alpha,z).
\]
For this choice of $f$, \cref{eq:jackson-coefficient-bound}
therefore gives $\sum_{k=1}^K k^2a_k^2\leq CL\,\TV(\alpha,z)$.
We will combine this coefficient bound with the bound on
the estimated moment discrepancies to prove
\Cref{cor:cheb-recovery,cor:buffered-recovery}.

\subsection{A common proof of the recovery guarantees}
\label{sec:common-recovery-proof}

We prove \Cref{cor:cheb-recovery,cor:buffered-recovery} together.
The optimal test function, the coefficient bound, and the
least-squares comparison are the same in the two arguments.
The approximation error is the only step where the calculations
differ.

\begin{proof}[Proof of \Cref{cor:cheb-recovery,cor:buffered-recovery}]
Let $q_k=\phi_k$ for all $k$, or $q_k=\psi_k$ for all $k$.
We first bound $\TV(\alpha,z)$ for an arbitrary candidate
$z\in\calZ_d$ in terms of its moment discrepancies.

\paragraph{Choosing the test function.}
Define the cumulative count difference
\[
 F(x)=\#\{i:\alpha_i\leq x\}-\#\{i:z_i\leq x\}.
\]
The cumulative representation of Wasserstein
distance and
\cref{eq:w1-list} give
\[
 \int_0^L|F(x)|\diff x
 =W_1(\mu_\alpha,\mu_z)=2 \TV(\alpha,z).
\]
Define
\[
 f(x) := -\int_0^x\operatorname{sgn}(F(t))\diff t,
 \qquad \operatorname{sgn}(0)=0.
\]
Then $f(0)=0$ and $f$ is $1$-Lipschitz, since
$|f'|\leq1$ almost everywhere.
Using $f(x)=f(L)-\int_x^L f'(t)\diff t$ and the definition of
$F$, we obtain
\begin{equation*}
 \sum_{i=1}^d\bigl(f(\alpha_i)-f(z_i)\bigr) =-\int_0^L f'(t)F(t)\diff t =\int_0^L|F(t)|\diff t=2\TV(\alpha,z).
\end{equation*}
Moreover, $F$ is integer-valued because both measures
are counting measures with unit weights. Hence
$|f'|^2=\indic{F\ne0}\leq|F|$ almost everywhere, and
\begin{equation}\label{eq:optimal-test-function-energy}
 \int_0^L|f'(x)|^2\diff x
 \leq\int_0^L|F(x)|\diff x=2\TV(\alpha,z).
\end{equation}

\paragraph{Applying the preceding bounds.}
Choose $P_f=\sum_{k=1}^K a_kq_k$ from
\Cref{lem:jackson-coefficient-energy}.
Together with \cref{eq:optimal-test-function-energy},
its coefficient bound gives
\begin{equation*}
 \sum_{k=1}^K k^2a_k^2\leq CL\TV(\alpha,z).
\end{equation*}
The approximation bounds summed over the two spectra in \cref{eq:full-approximation-error,eq:interior-approximation-error} in
\Cref{sec:approximation-error} give
\[
 \sum_{i=1}^d|f(\alpha_i)-P_f(\alpha_i)|
 +\sum_{i=1}^d|f(z_i)-P_f(z_i)|
 \leq CA,
 \qquad
 A=
 \begin{cases}
  \sqrt{dL}/K,&q_k=\phi_k,\\
  dL/K,&q_k=\psi_k.
 \end{cases}
\]
Recall that $\Delta_k(z)=\sum_i(q_k(\alpha_i)-q_k(z_i))$.

Applying \cref{eq:approximation-moment-decomposition}
and \cref{eq:weighted-coefficient-cauchy-schwarz} now yields
\begin{align*}
 2\TV(\alpha,z)
 &\leq CA+\left|\sum_{k=1}^K a_k\Delta_k(z)\right|\leq CA+C\sqrt{L\TV(\alpha,z)}
       \left(\sum_{k=1}^K\frac{\Delta_k(z)^2}{k^2}\right)^{1/2}\\
       &\leq CA+\TV(\alpha,z)
       +CL\sum_{k=1}^K\frac{\Delta_k(z)^2}{k^2}.
\end{align*}
The last step uses $2ab\leq a^2+b^2$.
Subtracting $\TV(\alpha,z)$ proves that, for every $z\in\calZ_d$,
\begin{equation}\label{eq:deterministic-recovery}
 \TV(\alpha,z)
 \leq CA+CL\sum_{k=1}^K\frac{\Delta_k(z)^2}{k^2}.
\end{equation}

\paragraph{Applying the bound to the estimator.}
Let $\widehat{\bm z}$ be the corresponding least-squares estimator.
Taking expectations in \cref{eq:deterministic-recovery} with
$z=\widehat{\bm z}$ gives
\[
 \E\TV(\alpha,\widehat{\bm z})
 \leq CA+CL\sum_{k=1}^K
       \frac{\E[\Delta_k(\widehat{\bm z})^2]}{k^2}
\leq CA+CL\sum_{k=1}^K
       \frac{\Var(\widehat{\bm F}_{q_k})}{k^2}.
\]
The last inequality uses the least-squares comparison in
\cref{eq:moment_expectation_error} and the variance identity in
\cref{eq:polynomial-estimation-variance}.
The choices $q_k=\phi_k$, $A=\sqrt{dL}/K$ and
$q_k=\psi_k$, $A=dL/K$ prove
\Cref{cor:cheb-recovery,cor:buffered-recovery}, respectively.
\end{proof}

\begin{remark}[Efficient Recovery] \label{rmk:poly-cheb-recovery}
The estimators in \Cref{cor:cheb-recovery,cor:buffered-recovery}
can be replaced by estimators with the same error guarantees,
up to universal constants and an arbitrarily small additive
tolerance $\eta>0$. Given the moment estimates, recovery takes
$\operatorname{poly}(d,K,1/\eta)$ time. Following the convex moment-fitting approach of \cite{MMRS25},
we optimize over nonnegative weights on a sufficiently fine grid,
subject to total mass $d$ and first moment at most $1$. While this may produce a solution with fractional weights, we can
round the cumulative weights to the nearest integers, obtaining a
list of exactly $d$ eigenvalues. This rounding, followed by rescaling
if their sum exceeds $1$, preserves the recovery guarantee up to
universal constant factors.
\end{remark}

\section{Variance of the estimated Chebyshev moments}
\label{sec:backends}

We bound the variances of the estimated Chebyshev moments needed
in the recovery guarantees of \Cref{sec:geometries}.

Condition on the bucketing output. 
Let $\overline\Pi$ be the small-bucket projector and write
$\sigma=\overline\Pi\rho\overline\Pi$, so $\tr(\sigma)\leq1$.
Fix $L>0$ such that $0\preceq\sigma\preceq LI_d$.
Let $n$ be the number of fresh input copies of $\rho$, including the discarded ones, used by either measurement scheme.

The unentangled and entangled measurement procedures described in
\Cref{sec:small-bucket-measurements} yield unbiased polynomial moment
estimates $\widehat{\bm F}_p$ satisfying the following variance bound.
We then apply it to $\phi_k$ and $\psi_k$.
We write $p^{(r)}$ for the $r$th derivative of the scalar polynomial $p$.

\begin{lemma}[Polynomial variance bound]
\label{lem:polynomial-variance}
Let $1\leq k\leq n$ and let $p$ be a polynomial of degree at most $k$.
For either measurement scheme,
\begin{equation}\label{eq:polynomial-variance}
 \Var(\widehat{\bm F}_p)
 \leq\sum_{r=1}^k\frac{c_r}{\binom nr}
       \left(\frac{\|p^{(r)}\|_{\infty,[0,L]}}{r!}\right)^2,
\end{equation}
where, for $r\geq1$,
\begin{equation}\label{eq:measurement-variance-factor}
 c_r:=\frac1{(r-1)!}
 \begin{cases}
  3^r d^{\uparrow r},
      &\text{unentangled},\\[4pt]
  \displaystyle\prod_{j=0}^{r-1}(1+jL),
      &\text{entangled}.
 \end{cases}
\end{equation}
\end{lemma}

The coefficients $c_r$ depend on the measurement scheme, but not
on $p$.  The difference between the bounds for $\phi_k$ and $\psi_k$
therefore comes from their derivatives.

\begin{corollary}[Chebyshev moment variance bounds]
\label{cor:explicit-chebyshev-variance}
For $1\leq k\leq n$ and either measurement scheme,
\begin{equation*}
 \begin{aligned}
 \Var(\widehat{\bm F}_{\phi_k})
 &\leq\sum_{r=1}^k\frac{c_r}{\binom nr}
       \left[\frac{(4k^2/L)^r}{(2r)!}\right]^2,\\
 \Var(\widehat{\bm F}_{\psi_k})
 &\leq\sum_{r=1}^k\frac{c_r}{\binom nr}
       \left[\frac{(2k/L)^r}{r!}\right]^2,
 \end{aligned}
\end{equation*}
where $c_r$ is given by \cref{eq:measurement-variance-factor}.
\end{corollary}

\begin{proof}
Substitute the derivative bounds from \Cref{lem:cheb-derivatives} into
\cref{eq:polynomial-variance}.
\end{proof}

In \Cref{sec:small-bucket-measurements}, we describe the two measurement procedures and their moment estimators.
We then prove \Cref{lem:polynomial-variance} in two steps.
First, in \Cref{sec:backend-snapshot,sec:backend-wss}, we use classical and quantum Efron--Stein decompositions to bound the variance in terms of the derivative operators of $A\mapsto\tr(p(A))$.
Second, in \Cref{sec:derivative-norms}, we bound these operators using the scalar derivatives $p^{(r)}$.

\subsection{Measurements on the small-bucket state}
\label{sec:small-bucket-measurements}

Fix $1\leq K\leq n$.
Both schemes first measure each of the $n$ fresh copies of $\rho$ with $\{\Pi,\overline\Pi\}$.
The outcome $\overline\Pi$ occurs with probability $\tr(\sigma)$ and, when $\tr(\sigma)>0$, leaves the normalized state $\sigma/\tr(\sigma)$ on the small-bucket subspace.
The complementary outcome $\Pi$ is discarded.

The unentangled and entangled schemes described below first produce real, unbiased moment estimates $\widehat{\bm{M}}_1,\ldots,\widehat{\bm{M}}_K$ satisfying
\[
 \E\widehat{\bm{M}}_m=\tr(\sigma^m),\qquad 1\leq m\leq K. 
\]
We then use them to construct a real, unbiased estimate $\widehat{\bF}_p$ of $\tr(p(\sigma))$ for any polynomial $p$ of degree at most $K$. 

\subsubsection{Unentangled measurements}
\label{subsub:unentangled-measurements}
We use the conditioned uniform POVM of \cite[Definition~5.9]{PTTW26}.
On each copy retained after outcome $\overline\Pi$, we apply the uniform POVM on $\C^d$.
Let $\ket{\bm u}\in\C^d$ be the resulting unit vector, and set
\[
 \bm X=(d+1)\ketbra{\bm u}{\bm u}-I_d.
\]
On a discarded copy, set $\bm X=0$.
Proposition~5.10 of \cite{PTTW26} gives $\E \bm{X}=\sigma$.  
For deterministic Hermitian matrices $X_1,\ldots,X_n$, define
\begin{equation*}
g_m(X_1,\ldots,X_n):=\frac1{\fall nm}
 \sum_{\substack{i_1,\ldots,i_m\in[n]\\\text{all distinct}}}
 \tr(X_{i_1}\cdots X_{i_m}),\qquad 1\leq m\leq K,
\end{equation*}
and set $g_0\equiv d$.  
Then by \cite[Definition~5.11]{PTTW26}, the monomial moment estimates are given by
\[
 \widehat{\bm M}_m=g_m(\bm X_1,\ldots,\bm X_n),\qquad 0\leq m\leq K,
\]
where $\bm X_1,\ldots,\bm X_n$ are i.i.d. estimators obtained by applying the conditioned POVM to the $n$ input copies.  
$\widehat{\bM}_1, \ldots, \widehat{\bM}_K$ are constructed using the same $\bm X_1,\ldots,\bm X_n$. 
For $p(x)=\sum_{m=0}^Kp_mx^m$, define
\begin{equation}\label{eq:snapshot-polynomial-function}
 g_p:=\sum_{m=0}^Kp_mg_m,
 \qquad
 \widehat{\bm F}_p=g_p(\bm X_1,\ldots,\bm X_n).
\end{equation}
The function $g_p$ is real-valued, symmetric, and affine in each argument.  
For every observable $A$,
\begin{equation}\label{eq:snapshot-diagonal-identity}
 g_p(A,\ldots,A)=\tr(p(A)).
\end{equation}
Independence and the fact that $g_p$ is affine in each argument give
\[
 \E\widehat{\bm F}_p
 =g_p(\sigma,\ldots,\sigma)
 =\tr(p(\sigma)).
\]

\subsubsection{Entangled measurements}
\label{subsub:entangled-measurements}

We use the moment-estimation procedure for a subnormalized state from \cite[Definition~5.5]{PSTW26}, which simultaneously measures the following observables.

Let $\overline\Pi_i$ denote $\overline\Pi$ acting on copy $i$ and the identity elsewhere.  
For $\pi\in S_r$, define the permutation operator $R_\pi$ on $r$ tensor factors by
\[
 R_\pi(\ket{v_1}\otimes\cdots\otimes \ket{v_r})
 =\ket{v_{\pi^{-1}(1)}}\otimes\cdots\otimes \ket{v_{\pi^{-1}(r)}},
 \qquad v_i\in [d].
\]
In particular, $R_{(i_1\cdots i_m)}$ cyclically permutes the indicated
copies and acts as the identity on the others. 
Define the following observables on the space $(\C^d)^{\otimes n}$ of all $n$ input copies by
\begin{equation*}
 \mathsf M_m:=\frac1{n^{\downarrow m}}
 \sum_{\substack{i_1,\ldots,i_m\in[n]\\\text{all distinct}}}
 \overline\Pi_{i_1}\cdots\overline\Pi_{i_m}R_{(i_1\cdots i_m)},
 \qquad 1\leq m\leq K,
\end{equation*}
and set $\mathsf M_0:=dI$.  
Each $\mathsf M_m$ is Hermitian, since every cycle occurs together with its inverse and commutes with the product of projectors on its support.  
Each $\mathsf M_m$ is also invariant under permutations of the copies.
Algorithmically, all $\mathsf M_m$ can be simultaneously measured via weak Schur sampling on the retained copies; see \cite[Definition~5.5 and Remark~3.11]{PSTW26}.
Let $\widehat{\bm M}_m$ denote the measured value of $\mathsf M_m$.
For every density matrix $\tau \in \mathrm{D}(\C^d)$, the trace identity for cycle permutations gives
\begin{equation}\label{eq:wss-moment-expectation}
 \tr(\mathsf M_m\tau^{\otimes n})
 =\tr\!\left((\overline\Pi\tau\overline\Pi)^m\right),
 \qquad 1\leq m\leq K.
\end{equation}

For a polynomial $p(x)=\sum_{m=0}^Kp_mx^m$ of degree at most $K$, set
\begin{equation*}
 \mathsf F_p:=\sum_{m=0}^Kp_m\mathsf M_m.
\end{equation*}
The polynomial moment estimate
\begin{equation*}
 \widehat{\bm F}_p=\sum_{m=0}^Kp_m\widehat{\bm M}_m
\end{equation*}
is therefore the measured value of $\mathsf F_p$.
By \cref{eq:wss-moment-expectation}, linearity, and $\mathsf M_0=dI$,
\begin{equation}\label{eq:wss-polynomial-expectation}
 \tr(\tau^{\otimes n}\mathsf F_p)
 =\tr\!\left(p(\overline\Pi\tau\overline\Pi)\right).
\end{equation}
Taking $\tau=\rho$ shows that $\widehat{\bm F}_p$ is an unbiased
estimator of $\tr(p(\sigma))$.

\subsection{Reducing the variance to derivative operators}

As shown in \Cref{sec:small-bucket-measurements}, the mean of $\widehat{\bm F}_p$ is $\tr(p(\sigma))$.
We bound its variance using derivatives of $A\mapsto\tr(p(A))$ at $A=\sigma$.

For $1\leq r\leq k$ and Hermitian $d\times d$ matrices $H_1,\ldots,H_r$, define
\[
 D_p^{(r)}(H_1,\ldots,H_r):=\left.
 \frac{\partial^r}{\partial t_1\cdots\partial t_r}
 \tr\!\left(p\!\left(\sigma+\sum_{j=1}^r t_jH_j\right)\right)
 \right|_{t_1=\cdots=t_r=0}.
\]
The matrices $H_1,\ldots,H_r$ are the directions in which we perturb $\sigma$, and $t_1,\ldots,t_r$ are real parameters.
We differentiate the trace once with respect to each parameter and then set all parameters to zero.
The result is a real number and is linear in each direction $H_j$.
Equivalently, $D_p^{(r)}(H_1,\ldots,H_r)$ is the coefficient of $t_1\cdots t_r$ in the polynomial $\tr(p(\sigma+\sum_{j=1}^r t_jH_j))$.

We represent $D_p^{(r)}$ by a matrix that stores its coefficients for all choices of directions $H_1,\ldots,H_r$. 
This allows us to bound the variance using Hilbert--Schmidt norms.
Define the Hermitian operator $\Theta_p^{(r)}$ on $(\C^d)^{\otimes r}$ by
\begin{equation}\label{eq:derivative-kernel}
 \tr\!\left[\Theta_p^{(r)}(H_1\otimes\cdots\otimes H_r)\right]
 :=D_p^{(r)}(H_1,\ldots,H_r).
\end{equation}
This identity determines $\Theta_p^{(r)}$ uniquely.
The operator $\Theta_p^{(r)}$ is invariant under permutations of the tensor factors because $D_p^{(r)}$ is symmetric in its arguments.

\begin{lemma}[Variance in terms of derivative operators]
\label{lem:derivative-variance}
Let $1\leq k\leq n$, and let $p$ be a polynomial of degree at most $k$.
For the conditioned uniform POVM,
\begin{equation}\label{eq:upovm-derivative-variance}
 \Var(\widehat{\bm{F}}_p)
 \leq\sum_{r=1}^k\frac{3^r}{r!\,n^{\downarrow r}}
       \left\|\Theta_p^{(r)}\right\|_{\mathrm{HS}}^2. 
\end{equation}
For weak Schur sampling,
\begin{equation}\label{eq:wss-derivative-variance}
 \Var(\widehat{\bm{F}}_p)
 \leq\sum_{r=1}^k\frac1{r!\,n^{\downarrow r}}
       \left\|\Theta_p^{(r)}(\sqrt{\sigma})^{\otimes r}\right\|_{\mathrm{HS}}^2.
\end{equation}
\end{lemma}

\subsubsection{Unentangled measurements}
\label{sec:backend-snapshot}

To prove \cref{eq:upovm-derivative-variance}, we use the following specialization of the Efron--Stein decomposition \cite[Theorem~8.35]{OD14}, also called the Hoeffding decomposition \cite[Section~11.4, p.~159]{vdV98}.
For completeness, we also include a proof. 

\begin{lemma}[Efron--Stein decomposition for symmetric multiaffine functions]
\label{lem:classical-efron-stein}
Let $\bm X_1,\ldots,\bm X_n$ be i.i.d. Hermitian matrices with mean $\sigma$ and finite second moments.  
Let $g$ be a real-valued symmetric function of $n$ such matrices that is affine in each argument.
Define 
\begin{equation*}
      f_g(A)\coloneqq g(A, \ldots, A). 
\end{equation*}
For any $1\leq r\leq n$ and Hermitian matrices $H_1,\ldots,H_r$, let
\begin{equation}\label{eq:classical_h_g}
h_g^{(r)}(H_1,\ldots,H_r) \coloneqq \frac{1}{n^{\downarrow r}} \frac{\partial^r}{\partial t_1\cdots\partial t_r} \left.
 f_g\!\left(\sigma+\sum_{j=1}^r t_jH_j\right)
 \right|_{t_1=\cdots=t_r=0}.
\end{equation}
Then
\[
 \Var(g(\bm X_1,\ldots,\bm X_n))
 =\sum_{r=1}^n \binom{n}{r}
   \E\left[h_g^{(r)}(\bm X_1-\sigma,\ldots,\bm X_r-\sigma)^2\right].
\]
\end{lemma}
\begin{proof}
Because $g$ is affine in each argument, its Taylor expansion around $(\sigma, \ldots, \sigma)$ contains only mixed derivatives in distinct arguments. 
When differentiating $f_g$ $r$ times, there are $n^{\downarrow r}$ ways to apply the derivatives to distinct arguments of $g$. 
By symmetry, every contribution equals $h_g^{(r)}$. 
Therefore, 
\begin{equation*}
     g(\bm X_1,\ldots,\bm X_n) = g(\sigma, \ldots, \sigma) + \sum_{S\subseteq [n], \, S\neq \varnothing} \underbrace{h_g^{(\abs{S})}\paren{(\bX_i - \sigma)_{i\in S}}}_{\coloneqq \bZ_S} . 
\end{equation*}
Each $\E \bZ_S = 0$ because it is linear in each input $\bX_i - \sigma$, $\E \bX_i = \sigma$, and all $\bX_i$ are independent. 
For any $S\neq T$, there exists an index $i$ such that $i$ is contained in one but not the other. Because $\bX_i$ is independent from the rest, we have $\E[\bZ_S \bZ_T] = 0$. Therefore $\Var(g(\bX_1, \ldots, \bX_n)) =  \sum_{S\subseteq [n], \, S\neq \varnothing} \E[\bZ_S^2]$. 
By the i.i.d. assumption and symmetry, all subsets $S$ of size $r$ contribute the same quantity. Grouping them proves the lemma. 
\end{proof}

\begin{proof}[Proof of \cref{eq:upovm-derivative-variance}]
Let $\bm X_1,\ldots,\bm X_n$ be i.i.d. estimators obtained as described in the unentangled procedure in \Cref{sec:small-bucket-measurements}, which satisfy $\E\bm X_i=\sigma$.
Then $\widehat{\bm F}_p=g_p(\bm X_1,\ldots,\bm X_n)$, where $\widehat{\bm F}_p$ and $g_p$ are given by \cref{eq:snapshot-polynomial-function}. 
Since $g_p$ is real-valued, symmetric, and affine in every argument, we can apply \Cref{lem:classical-efron-stein} to $g = g_p$ and obtain
\begin{equation}\label{eq:snapshot-efron-stein-variance}
 \Var(\widehat{\bm F}_p)
 =\sum_{r=1}^n\binom nr
   \E\left[h_{g_p}^{(r)}(\bm X_1-\sigma,\ldots,\bm X_r-\sigma)^2\right],
\end{equation}
where $h_{g_p}^{(r)}$ is defined in \cref{eq:classical_h_g}.
It follows from \cref{eq:snapshot-diagonal-identity} that $f_{g_p}(A) = g_p(A,\ldots,A)=\tr(p(A))$.
Therefore, 
\begin{equation}\label{eq:snapshot-coefficient-identity}
 n^{\downarrow r} h_{g_p}^{(r)}(H_1,\ldots,H_r)
 = D_p^{(r)}(H_1,\ldots,H_r) = \tr\!\left[\Theta_p^{(r)}(H_1\otimes\cdots\otimes H_r)\right],
\end{equation}
where the last equality follows from \cref{eq:derivative-kernel}. 
For $r>k$, we have $h_{g_p}^{(r)}=0$ because $p$ has degree at most $k$. 
Substituting \cref{eq:snapshot-coefficient-identity} into \cref{eq:snapshot-efron-stein-variance} and using $\binom nr=n^{\downarrow r}/r!$ gives
\begin{align*}
 \Var(\widehat{\bm{F}}_p)
 &=\sum_{r=1}^k\frac1{r!\,n^{\downarrow r}}
   \E\left[\tr\!\left[
     \Theta_p^{(r)}\paren[\big]{(\bm X_1-\sigma) \otimes \cdots \otimes (\bm X_r-\sigma)}
     \right]^2\right] \\
&\leq \sum_{r=1}^k\frac1{r!\,n^{\downarrow r}}
   \E\left[\tr\!\left[
     \Theta_p^{(r)}\paren[\big]{\bm X_1 \otimes \cdots \otimes \bm X_r}
     \right]^2\right], 
\end{align*}
where the inequality is because centering any $\bX_i$ by subtracting its mean doesn't increase the second moment. 
Applying \Cref{lem:snapshot-covariance}
below with $A=\Theta_p^{(r)}$ completes the proof.
\end{proof}

We prove the following using the second-moment formula in
\cite[Proposition~5.10]{PTTW26} and the independence of $\bX_1, \ldots, \bX_r$.

\begin{lemma}[Conditioned uniform POVM second-moment bound]
\label{lem:snapshot-covariance}
Let $r\geq1$ and let $\bm X_1,\ldots,\bm X_r$
be i.i.d. estimators obtained from the conditioned uniform POVM
in \Cref{sec:small-bucket-measurements}.
For every Hermitian operator $A$ on $(\C^d)^{\otimes r}$,
\begin{equation}\label{eq:snapshot-covariance}
 \E\left[\tr\!\left(
 A\paren[\big]{\bm X_1\otimes\cdots\otimes\bm X_r}
 \right)^2\right]
 \leq 3^r\norm{A}_{\mathrm{HS}}^2.
\end{equation}
\end{lemma}
\begin{proof}
We argue by induction on $r$. For $r=1$, 
\begin{align*}
      (d+2)\E[\tr(A\bm X_1)^2] &= (d+2) \tr\left[(A \otimes A) \E[\bX_1 \otimes \bX_1]\right] \\
      &= \tr \left[ (A \otimes A) \cdot
            ((d+1) \SWAP - I) \cdot (\tr(\sigma) I \otimes I + \sigma \otimes I + I \otimes \sigma)
      \right] \tag{\cite[Proposition~5.10]{PTTW26}} \\
      &= (d+1)\tr(\sigma) \tr(A^2)+(2d+3)\tr(\sigma A^2) - \tr[\sigma(A + \tr(A)I)^2] \\
      &\leq (3d+4) \norm{A}_{\mathrm{HS}}^2, 
\end{align*}
where in the last inequality we used $\tr(\sigma)\leq1$ and $0\preceq \sigma\preceq I_d$.  
The bound $(3d+4)/(d+2)\leq 3$ therefore proves
\cref{eq:snapshot-covariance} for $r=1$.

Now consider $r \geq 2$. 
Denote by $\{E_a\}$ a Hilbert--Schmidt orthonormal basis of Hermitian $d\times d$ matrices such that $\tr(E_a E_b) = \delta_{ab}$. 
Then every Hermitian matrix $A$ can be written as $A=\sum_a E_a\otimes A_a$ for some Hermitian matrices $A_a$ on the remaining $r-1$ tensor factors.
Write $\bZ = \bm X_2\otimes\cdots\otimes\bm X_r$. 
Then
\begin{equation*}
\begin{aligned}
\tr\!\left[A(\bm X_1\otimes\bZ)\right]
&=\sum_a\tr\!\left[(E_a\bm X_1)\otimes \bigl(A_a\bZ\bigr)\right]\\
&=\sum_a\tr(E_a\bm X_1) \tr(A_a\bZ)
= \tr \left[\paren*{\sum_a \tr(A_a\bZ) E_a} \bX_1\right].
\end{aligned}
\end{equation*}
Since $\bm Z$ is independent of $\bm X_1$, we may condition on
$\bm Z$ and apply the base case to the Hermitian matrix
$\sum_a\tr(A_a\bm Z)E_a$.
Therefore,
\begin{align*}
\E \left[\tr\!\left[A(\bm X_1\otimes\bZ)\right]^2 \right] &= \E_{\bZ} \E_{\bX_1} \left[ \tr \left[\paren*{\sum_a \tr(A_a\bZ) E_a} \bX_1\right]^2 \right] \\
&\leq 3 \E_{\bZ} \tr \left[\paren*{\sum_a \tr(A_a\bZ) E_a}^2\right] \tag{base case of $r=1$} \\
&= 3 \sum_a \E_{\bZ} \tr(A_a \bZ)^2 \tag{$\tr(E_a E_b) = \delta_{ab}$} \\
&\leq 3 \sum_a 3^{r-1} \norm{A_a}_{\mathrm{HS}}^2 \tag{induction hypothesis} \\
&= 3^r \norm{A}_{\mathrm{HS}}^2 \tag{Hilbert-Schmidt norm orthogonality}, 
\end{align*}
which completes the proof. 
\end{proof}

\subsubsection{Entangled measurements}
\label{sec:backend-wss}

To prove \cref{eq:wss-derivative-variance}, we use the following
variant of the quantum Efron--Stein decomposition in
\cite[Theorem~6.9 and Proposition~6.11]{PFMO25}.
The decomposition also appears under the name quantum Hoeffding
decomposition in \cite[Theorem~5.6]{GB10}.
For completeness, we also include a proof. 

\begin{lemma}[Variance bound from the quantum Efron--Stein decomposition]
\label{lem:quantum-efron-stein}
Let $\rho\in \mathrm{D}(\calH)$ for some Hilbert space
$\mathcal H$, and let $G$ be an observable on $\mathcal H^{\otimes n}$
invariant under permutations of the tensor factors.
Define 
\begin{equation*}
      f_G(A)\coloneqq \tr(G A^{\otimes n}).
\end{equation*}
For any $1\leq r\leq n$, choose a permutation-invariant Hermitian operator $K_r$ on $\mathcal H^{\otimes r}$ satisfying
\begin{equation}\label{eq:K_r}
\tr \left[K_r (H_1 \otimes \cdots \otimes H_r)\right] = \frac{1}{n^{\downarrow r}} \frac{\partial^r}{\partial t_1\cdots\partial t_r} \left. f_G\!\left(\rho+\sum_{j=1}^r t_jH_j\right) \right|_{t_1=\cdots=t_r=0}
\end{equation}
for all traceless Hermitian $H_1, \ldots, H_r$. 
Then the variance of measuring $G$ in state $\rho^{\otimes n}$ satisfies
\begin{equation*}
 \tr(G^2\rho^{\otimes n})-\tr(G\rho^{\otimes n})^2
 \leq\sum_{r=1}^n\binom nr \tr\paren*{K_r^2 \rho^{\otimes r}}. 
\end{equation*}
\end{lemma}

\begin{proof}
For any density matrix $\tau$, the difference $\Delta:=\tau-\rho$
is traceless. Since $f_G$ is a polynomial of degree at most $n$,
Taylor's formula and the definition of $K_r$ give
\begin{equation}\label{eq:quantum-coefficient-expansion}
 \begin{aligned}
 \tr(G\tau^{\otimes n})
 &=\tr(G\rho^{\otimes n}) + \sum_{r=1}^n\frac1{r!}
   \left.\frac{\mathrm d^r}{\mathrm dt^r}
   f_G(\rho+t\Delta)\right|_{t=0}=\tr(G\rho^{\otimes n}) + \sum_{r=1}^n\binom nr
      \tr\!\left[K_r(\tau-\rho)^{\otimes r}\right],
 \end{aligned}
\end{equation}
where we used $n^{\downarrow r}/r!=\binom nr$. 
Define $\mathcal C_\rho(A):=A-\tr(\rho A)I$ and $\widetilde K_r:=\mathcal C_\rho^{\otimes r}(K_r)$. 
We now show that \cref{eq:quantum-coefficient-expansion} gives a useful decomposition of $G$: 
\begin{equation}\label{eq:decompose_G}
      G=\tr(G\rho^{\otimes n})I^{\otimes n} + \sum_{S\subseteq[n], \, S\neq \varnothing} \widetilde K_{|S|}^{(S)}, 
\end{equation}
where $\widetilde K_{|S|}^{(S)}$ acts on the factors in $S$ as $\widetilde K_{|S|}$ and as
the identity elsewhere.
Since both sides of \cref{eq:decompose_G} are permutation-invariant, it suffices to check that $\tr(G \tau^{\otimes n}) = \tr(\mathrm{RHS} \cdot \tau^{\otimes n})$ for every density operator $\tau$. 
Indeed,
\begin{align*}
      \tr(\mathrm{RHS}\cdot \tau^{\otimes n}) &= \tr(G\rho^{\otimes n}) + \sum_{r=1}^n \binom{n}{r} \tr\paren*{\widetilde{K}_r \tau^{\otimes r}} \\
      &= \tr(G\rho^{\otimes n}) + \sum_{r=1}^n\binom nr
      \tr\!\left[K_r(\tau-\rho)^{\otimes r}\right] = \tr(G \tau^{\otimes n}). 
\end{align*}

Write $Z_S \coloneqq \widetilde K_{|S|}^{(S)}$. 
Each $Z_S$ has zero mean when any factor in $S$ is averaged against $\rho$. 
For $S\ne T$, choose an index belonging to exactly one of
the two subsets. Averaging over this factor gives
$\tr(\rho^{\otimes n}Z_SZ_T)=0$.
Therefore, taking the variance and grouping subsets of the same size gives
\begin{align*}
 \tr(G^2\rho^{\otimes n})-\tr(G\rho^{\otimes n})^2
 =\sum_{S\subseteq[n], \, S\ne\varnothing}
      \tr(\rho^{\otimes n}Z_S^2)=\sum_{r=1}^n\binom nr
      \tr\paren*{\widetilde K_r^2\rho^{\otimes r}}.
\end{align*}
Since centering each tensor factor with respect to $\rho$ cannot increase the second moment, we have $\tr(\widetilde K_r^2\rho^{\otimes r})
 \leq \tr(K_r^2\rho^{\otimes r}),$
which proves the claim.
\end{proof}

\begin{proof}[Proof of \cref{eq:wss-derivative-variance}]
Let $G=\mathsf F_p$, the permutation-invariant observable whose measured
value is $\widehat{\bm F}_p$.
By \cref{eq:wss-polynomial-expectation},
$f_G(A)=\tr(G A^{\otimes n}) = \tr(p(\overline\Pi A\overline\Pi))$
for every trace-one Hermitian matrix $A$.
Choose
\[
 K_r=\frac1{n^{\downarrow r}}\,
     \overline\Pi^{\otimes r}\Theta_p^{(r)}\overline\Pi^{\otimes r},
 \quad 1\leq r\leq k,
 \qquad K_r=0,\quad k<r\leq n.
\]
These operators satisfy \cref{eq:K_r} by \cref{eq:derivative-kernel},
since derivatives of order greater than $k$ vanish.
Applying \Cref{lem:quantum-efron-stein} and using $\sigma=\overline\Pi\rho\overline\Pi$
and $0\preceq\overline\Pi^{\otimes r}\preceq I$ gives
\begin{align*}
 \Var(\widehat{\bm F}_p)
 &\leq\sum_{r=1}^k\binom nr
       \tr(K_r^2\rho^{\otimes r})=\sum_{r=1}^k\frac1{r!\,n^{\downarrow r}}
       \tr\!\left[
                   \overline\Pi^{\otimes r}\Theta_p^{(r)} \sigma^{\otimes r}\Theta_p^{(r)}\right]\\
&\leq \sum_{r=1}^k\frac1{r!\,n^{\downarrow r}}
       \tr\!\left[
                   \Theta_p^{(r)} \sigma^{\otimes r}\Theta_p^{(r)}\right]=\sum_{r=1}^k\frac1{r!\,n^{\downarrow r}}
       \left\|\Theta_p^{(r)}
                   (\sqrt\sigma)^{\otimes r}\right\|_{\mathrm{HS}}^2. \qedhere
\end{align*}
\end{proof}

\subsection{Bounding the derivative operators}
\label{sec:derivative-norms}

It remains to bound $\|\Theta_p^{(r)}\|_{\mathrm{HS}}^2$ and $\|\Theta_p^{(r)}(\sqrt{\sigma})^{\otimes r}\|_{\mathrm{HS}}^2$ in \Cref{lem:derivative-variance}.
This subsection uses only matrix analysis.
The bounds in \Cref{lem:derivative-norms} apply to any $0\preceq\sigma\preceq LI_d$ with $\tr(\sigma)\leq1$ and do not depend on the measurement scheme.

The following lemma is the matrix-derivative formula with a bound on the coefficients. We prove it in \Cref{app:trace-derivatives}.

\begin{restatable}[Trace derivative formula]{lemma}{tracederivatives}
\label{lem:trace-derivatives}
Let $p$ be a real polynomial of degree at most $k$, and let
$\sigma=\operatorname{diag}(\lambda_1,\ldots,\lambda_d)$ with
$\lambda_i\in[0,L]$.
For every $1\leq r\leq k$, there are real coefficients
$w_{i_1,\ldots,i_r}$, symmetric in their indices, satisfying
\[
 |w_{i_1,\ldots,i_r}|
 \leq \|p^{(r)}\|_{\infty,[0,L]},
\]
such that, for all Hermitian $H_1,\ldots,H_r\in\C^{d\times d}$,
\begin{equation}\label{eq:trace-derivative-formula}
 D_p^{(r)}(H_1,\ldots,H_r)
 =\frac{1}{(r-1)!}
 \sum_{\substack{\pi\in S_r\\\pi(1)=1}}
 \sum_{i_1,\ldots,i_r=1}^d
 w_{i_1,\ldots,i_r}
 \prod_{j=1}^r [H_{\pi(j)}]_{i_j i_{j+1}},
\end{equation}
where $i_{r+1}=i_1$.
The condition $\pi(1)=1$ means that we sum over all orders of
$H_2,\ldots,H_r$, keeping $H_1$ first.
\end{restatable}

\begin{example}
For $r=2$, the formula reads
\begin{equation*}
 D_p^{(2)}(H_1,H_2)
 =\sum_{i,j=1}^d w_{ij}[H_1]_{ij}[H_2]_{ji},
 \qquad
 w_{ij}=\int_0^1
 p''\bigl((1-t)\lambda_i+t\lambda_j\bigr)\diff t.
\end{equation*}
Thus $w_{ij}$ is an average of $p''$ between the two eigenvalues,
so $|w_{ij}|\leq\|p''\|_{\infty,[0,L]}$.
\end{example}

\begin{lemma}[Bounds for the derivative operators]
\label{lem:derivative-norms}
Let $p$ be a polynomial of degree at most $k$, and let
$0\preceq\sigma\preceq LI_d$ with $\tr(\sigma)\leq1$.
For every $1\leq r\leq k$, the derivative operator $\Theta_p^{(r)}$ in
\cref{eq:derivative-kernel} satisfies
\begin{align*}
 \|\Theta_p^{(r)}\|_{\mathrm{HS}}^2
 &\leq\frac{d^{\uparrow r}}{(r-1)!}
 \norm{p^{(r)}}_{\infty,[0,L]}^2,\\
 \left\|\Theta_p^{(r)}(\sqrt{\sigma})^{\otimes r}\right\|_{\mathrm{HS}}^2
 &\leq\frac{\prod_{j=0}^{r-1}(1+jL)}{(r-1)!}
 \norm{p^{(r)}}_{\infty,[0,L]}^2.
\end{align*}
\end{lemma}

\begin{proof}
Let us work in an eigenbasis of $\sigma$.
Let $w_{i_1,\ldots,i_r}$ be the coefficients in
\Cref{lem:trace-derivatives}, and define $W_r$ to be diagonal in
the tensor-product basis, with diagonal entry $w_{i_1,\ldots,i_r}$
at $\ket{i_1}\otimes\cdots\otimes\ket{i_r}$.
Then
\begin{equation}\label{eq:HG}
 \|W_r\|_\infty\leq\|p^{(r)}\|_{\infty,[0,L]}.
\end{equation}

Let $\calC_r$ be the set of $r$-cycles on $[r]$.  We use the
permutation operators $R_\pi$ defined in \Cref{sec:small-bucket-measurements},
acting here on $r$ tensor factors.
Since the coefficients $w_{i_1,\ldots,i_r}$ are symmetric,
the permutation-trace formula in \cite[Lemma~5.6]{PTTW26}
identifies each ordering in \cref{eq:trace-derivative-formula}
with one $r$-cycle starting at $1$.
Every $r$-cycle occurs once, so
\begin{equation}\label{eq:divdiff-rep}
 \Theta_p^{(r)}
 =\frac1{(r-1)!}W_r\sum_{c\in\calC_r}R_c.
\end{equation}
Using \cref{eq:divdiff-rep},
$\|W_rM\|_{\mathrm{HS}}\leq\|W_r\|_\infty\|M\|_{\mathrm{HS}}$,
and $|\calC_r|=(r-1)!$, we have
\begin{equation*}
 \|\Theta_p^{(r)}\|_{\mathrm{HS}}^2
 \leq\frac{\|W_r\|_\infty^2}{((r-1)!)^2}
       \sum_{c,c'\in\calC_r}
       \tr(R_{c^{-1}c'}) \leq \frac{\|W_r\|_\infty^2}{(r-1)!}
       \sum_{\pi\in S_r}
       \tr(R_\pi).
\end{equation*}
In the last inequality, for each fixed $c\in\calC_r$, the permutations $c^{-1}c'$ are
distinct as $c'$ varies. Since the traces are nonnegative,
we extend the inner sum to all of $S_r$.
Similarly,
\begin{equation*}
 \left\|\Theta_p^{(r)}(\sqrt{\sigma})^{\otimes r}\right\|_{\mathrm{HS}}^2
 \leq\frac{\|W_r\|_\infty^2}{(r-1)!}
       \sum_{\pi\in S_r}
       \tr(\sigma^{\otimes r}R_\pi).
\end{equation*}

The permutation-trace formula in \cite[Lemmas~5.6--5.7]{PTTW26}
and the assumptions $0\preceq\sigma\preceq LI_d$ and
$\tr(\sigma)\leq1$ give, for every $\pi\in S_r$ with $\#\pi$ cycles,
\[
 \tr(R_\pi)=d^{\#\pi},
 \qquad
 0\leq\tr(\sigma^{\otimes r}R_\pi)\leq L^{r-\#\pi}.
\]
Therefore
\begin{align*}
 \sum_{\pi\in S_r}\tr(R_\pi)
 &=\sum_{\pi\in S_r}d^{\#\pi}=d^{\uparrow r},\\
 \sum_{\pi\in S_r}\tr(\sigma^{\otimes r}R_\pi)
 &\leq\sum_{\pi\in S_r}L^{r-\#\pi}
 =\prod_{j=0}^{r-1}(1+jL).
\end{align*}
Combining these estimates with \cref{eq:HG} proves both bounds.
\end{proof}

\begin{proof}[Proof of \Cref{lem:polynomial-variance}]
Substitute \Cref{lem:derivative-norms} into \Cref{lem:derivative-variance}. 
\end{proof}

\section{Algorithms and analysis}
\label{sec:rates}

We first give a full description of our entangled and unentangled spectrum estimation algorithms, then combine the recovery and variance bounds to prove our main results \Cref{thm:entangled,thm:unentangled}.
We describe both entangled and unentangled spectrum estimation algorithms in \Cref{alg:spectrum-estimation} as they follow a common procedure.  
The measurement model determines the bucketing algorithm and the small-bucket measurements; the accuracy regime determines the threshold $B$, the degree $K$, and the polynomials whose moments we estimate.
Fix $d\geq3$ and an accuracy parameter $0<\varepsilon\leq1/10$.
Let $b>0$ be a sufficiently small universal constant.
Set $\ell=\log d$.
We choose the other parameters from
\Cref{tab:spectrum-models,tab:spectrum-regimes}. 

\begin{table}[H]
\centering
\renewcommand{\arraystretch}{1.4}
\begin{tabularx}{\textwidth}{@{}lYY@{}}
\toprule
 & Entangled & Unentangled\\
\midrule
Bucketing algorithm
 & \Cref{prop:bucketing-entangled} & \Cref{prop:bucketing-single}\\
Bucketing copies $N$
 & $O\!\left(d/(B\varepsilon^2)\right)$
 & $O\!\left(d/(B^2\varepsilon^2)\right)$\\
Small-bucket measurements
 & \Cref{subsub:entangled-measurements}
 & \Cref{subsub:unentangled-measurements}\\
Fresh copies $n$
 & $O\!\left(d/(B\varepsilon^2)\right)$
 & $O\!\left(d^2/(B\varepsilon^2)\right)$\\
\bottomrule
\end{tabularx}
\caption{Choices and parameters determined by the measurement model.  
The fresh copy counts $n$ for the small-bucket measurements include discarded inputs.}
\label{tab:spectrum-models}
\end{table}

\begin{table}[H]
\centering
\renewcommand{\arraystretch}{1.5}
\begin{tabular*}{\textwidth}{@{\extracolsep{\fill}}llcc@{}}
\toprule
$\varepsilon$ regime & Method: Chebyshev moment matching & Degree $K$ & Threshold $B$\\
\midrule
$\varepsilon\ell>1$
 & On the full interval ($p_k=\phi_k$)
 & $\lceil \ell^2\rceil$ & $b\varepsilon^2K^2/(1.1d)$\\
$\varepsilon\ell\leq1$
 & In the interior ($p_k=\psi_k$)
 & $\lceil \ell^2/\varepsilon\rceil$ & $\ell^2/(1.1d)$\\
\bottomrule
\end{tabular*}
\caption{Choices and parameters determined by the accuracy regime.  
The polynomials $\phi_k$ and $\psi_k$ are defined in \cref{eq:two-chebyshev-phi,eq:two-chebyshev-psi}, with $L = 1.1B$.}
\label{tab:spectrum-regimes}
\end{table}

\begin{mdframed}[nobreak=true,linewidth=0.6pt,innertopmargin=4pt,innerbottommargin=8pt]
\begin{algorithm}[Spectrum estimation]
\label{alg:spectrum-estimation}
\textbf{Input:} $N+n$ copies of an unknown state $\rho$ on $\C^d$, accuracy
$0<\varepsilon\leq1/10$, and a choice of entangled or unentangled measurements.

Choose the parameters in \Cref{tab:spectrum-models,tab:spectrum-regimes}.

\begin{enumerate}[leftmargin=*,label=\arabic*.]
\item \textbf{Bucket.}
Use $N$ copies to run the corresponding bucketing algorithm at threshold $B$ and accuracy $\varepsilon$.
Obtain $\bm\Pi$, $\overline{\bm\Pi}$, and $\widehat{\bm\alpha}_{\mathrm{Large}}$.

\item \textbf{Estimate moments.}
Use $n$ fresh copies and $\{\bm\Pi,\overline{\bm\Pi}\}$ to perform the corresponding small-bucket measurements from \Cref{sec:small-bucket-measurements}.  
Compute $\widehat{\bm F}_{p_k}$ for every $1\leq k\leq K$.

\item \textbf{Reconstruct the small spectrum.}
Choose
\[
 \widehat{\bm z}\in
 \argmin_{\substack{\sum_i z_i\leq1, \\ 1.1B\geq z_1\geq\cdots\geq z_d\geq 0}} \quad 
 \sum_{k=1}^K\frac1{k^2}
 \left(\widehat{\bm F}_{p_k}-\sum_{i=1}^d p_k(z_i)\right)^2.
\]
\end{enumerate}

\textbf{Output:} 
Keep the largest $\rank(\overline{\bm\Pi})$ entries of $\widehat{\bm z}$ and concatenate them with $\widehat{\bm\alpha}_{\mathrm{Large}}$. Return the sorted list.
\end{algorithm}
\end{mdframed}

\subsection{\texorpdfstring{Proof of \Cref{thm:entangled,thm:unentangled}}{Proof of the main theorems}}

Write $D=1$ for entangled measurements and $D=d$ for unentangled measurements.
Use the parameters in
\Cref{tab:spectrum-models,tab:spectrum-regimes}, with
$L=1.1B$.
Set
\[
 n=\left\lceil A\frac{dD}{L\varepsilon^2}\right\rceil,
\]
where $A$ is a universal constant large enough for the variance
bounds below.
Let $b>0$ be a sufficiently small universal constant.

Condition on any successful bucketing output, so that
$0\preceq\bm\sigma\preceq LI_d$ and $\tr(\bm\sigma)\leq1$.
We first show that the conditional expected error of
$\widehat{\bm z}$ is $O(\varepsilon)$.

\paragraph{Full interval: $\varepsilon\ell>1$.}
Here $L=b\varepsilon^2K^2/d$.
Since $K=O(\ell^2)$, we have $L=O(b)$ and $K/n=O(b/A)$.
Choosing $b$ sufficiently small ensures $B=L/1.1<1$,
as required by the bucketing guarantees, and taking
$A$ sufficiently large ensures $K\leq n/2$,
so that $\binom nr^{-1}\leq r!(2/n)^r$ for $1\leq r\leq K$.

We now show that the coefficients $c_r$ defined in \cref{eq:measurement-variance-factor}
satisfy $c_r\leq C(6D)^r/(r-1)!$ for $1\leq r\leq K$
in both measurement models.
For unentangled measurements, $r\leq K\leq d$ gives
\[
 d^{\uparrow r}=\prod_{j=0}^{r-1}(d+j)\leq(2d)^r.
\]
For entangled measurements, $1+x\leq e^x$ gives
\[
 \prod_{j=0}^{r-1}(1+jL)
 \leq e^{Lr(r-1)/2}
 \leq e^{LK^2/2}
 =O(1),
\]
where we used $LK^2=b\varepsilon^2K^4/d=O(b)$.
Substituting these bounds into \cref{eq:measurement-variance-factor} gives the claimed bound on $c_r$. 

Using $r/((2r)!)^2\leq32^r/(4r)!$,
\Cref{cor:explicit-chebyshev-variance} gives
\begin{equation*}
 \Var(\widehat{\bm F}_{\phi_k})
 \leq C\sum_{r\geq1}
       \frac{(CDk^4/(nL^2))^r}{(4r)!}
 \leq C\exp\!\left(
       C\left(\frac{DK^4}{nL^2}\right)^{1/4}\right)
 \leq Cd^{1/4}.
\end{equation*}
The last inequality holds for sufficiently large $A$, since
$DK^4/(nL^2)\leq K^2/(Ab)$ and $\sqrt K=O(\ell)$.
Since $\sum_{k=1}^K k^{-2}=O(1)$,
\Cref{cor:cheb-recovery} yields
\[
 \E\TV(\widehat{\bm z},\spec(\bm\sigma))
 \leq C\frac{\sqrt{dL}}K+CLd^{1/4}
 =O\!\left(\sqrt b\,\varepsilon+
       b\varepsilon^2\frac{\ell^4}{d^{3/4}}\right)
 =O(\varepsilon).
\]

\paragraph{Interior: $\varepsilon\ell\leq1$.}
Here $L=\ell^2/d$ and $K=\lceil\ell^2/\varepsilon\rceil$.
Then $B=L/1.1<1$,
as required by the bucketing guarantees.
Also, $K/n=O(A^{-1})$, so taking $A$ sufficiently large ensures
$K\leq n/2$ and hence $\binom nr^{-1}\leq r!(2/n)^r$
for $1\leq r\leq K$.

We now show that the coefficients $c_r$ defined in
\cref{eq:measurement-variance-factor} satisfy
\[
 c_r\leq
 \frac{(3D)^r}{(r-1)!}(1+(r-1)L)^{r-1}
\]
for $1\leq r\leq K$ in both measurement models.
For unentangled measurements, $L\geq1/d$ gives
\[
 d^{\uparrow r}
 =d^r\prod_{j=0}^{r-1}(1+j/d)
 \leq d^r(1+(r-1)L)^{r-1}.
\]
For entangled measurements,
\[
 \prod_{j=0}^{r-1}(1+jL)
 \leq(1+(r-1)L)^{r-1}.
\]
Substituting these bounds into \cref{eq:measurement-variance-factor}
gives the claimed bound on $c_r$.

Set $t=24DK^2/(nL^2)$.
Using this coefficient bound and setting $j=r-1$,
\Cref{cor:explicit-chebyshev-variance} gives
\begin{equation*}
 \frac{\Var(\widehat{\bm F}_{\psi_k})}{k^2}
 \leq\frac{24D}{nL^2}
       \sum_{j\geq0}\frac{t^j(1+jL)^j}{j!(j+1)!}
 \leq\frac{24D}{nL^2}
       \left(e^{2\sqrt{2t}}+e^{2eLt}\right)
 \leq C\frac{D}{nL^2}d^{1/4}.
\end{equation*}
The second inequality follows by applying
$(1+jL)^j\leq2^j(1+(jL)^j)$ and $j^j\leq e^j j!$
for $j\geq1$.
The last inequality holds for sufficiently large $A$, since
$t=O(\ell^2/A)$ and $Lt=O(\ell^4/(Ad))=O(A^{-1})$.
Summing this bound over $1\leq k\leq K$ and applying
\Cref{cor:buffered-recovery} yields
\[
 \E\TV(\widehat{\bm z},\spec(\bm\sigma))
 \leq C\frac{dL}{K}
       +C\frac{KD}{nL}d^{1/4}
 =O\!\left(\varepsilon+
       \frac{\varepsilon\ell^2}{Ad^{3/4}}\right)
 =O(\varepsilon).
\]

\paragraph{Accuracy.}
Bucketing succeeds with probability at least $0.99$.
Conditional on any successful output, Markov's inequality gives $\TV(\widehat{\bm z},\spec(\bm\sigma))=O(\varepsilon)$ with probability at least $0.9$.
Let $\bm x$ be the output of \Cref{alg:spectrum-estimation}.
Hence \Cref{item:pack-list,item:pack-pinch} give
\[
 \TV(\bm x,\spec(\rho))
 \leq 2\varepsilon+\TV(\widehat{\bm z},\spec(\bm\sigma))
 =O(\varepsilon)
\]
with probability at least $0.99\cdot0.9>0.8$.
Standard success amplification raises this probability to $0.99$.

\paragraph{Copy complexity.}
The bucketing cost in \Cref{tab:spectrum-models} is $O(n)$
for entangled measurements and $O(n/(dB)) = O(n)$ for unentangled measurements. 
Substituting the two choices of $L$ gives
\[
 O\!\left(
 d^2D\min\left\{
   \frac1{(\varepsilon\ell)^2},
   \frac1{(\varepsilon\ell)^4}
 \right\}\right).
\]
Taking $D=1$ and $D=d$ proves
\Cref{thm:entangled,thm:unentangled}, respectively.

\section*{Acknowledgments}
We thank Christopher Musco for several insightful discussions. 
AB is grateful to Ewin Tang and John Wright for helpful discussions. Part of this work was done while AB was visiting the Simon's Institute. XT is supported by Phyllis Ruby Block Fellowship. 

\paragraph{AI Disclosure.} Chebyshev moment matching is a relatively new framework developed in the numerical linear algebra community. The idea to extend this framework to the quantum setting was generated anthropically (i.e., by human). The authors used GPT 5.6 Sol to assist with the technical proofs and GPT 6 Astra to assist with the technical writing. The authors take full responsibility for the content and correctness of this work.

\printbibliography

@InProceedings{PTTW26,
author = {Angelos Pelecanos and Xinyu Tan and Ewin Tang and John Wright},
title = {Beating full state tomography for unentangled spectrum estimation},
booktitle = {Proceedings of the 2026 Annual ACM-SIAM Symposium on Discrete Algorithms (SODA)},
chapter = {},
pages = {3313-3363},
year = {2026},
eprinttype = {arXiv},
primaryClass={quant-ph},
eprint = {2504.02785}
}

@InProceedings{MMRS25,
  title = 	 {Sharper Bounds for Chebyshev Moment Matching, with Applications},
  author =       {Musco, Cameron and Musco, Christopher and Rosenblatt, Lucas and Singh, Apoorv Vikram},
  booktitle = 	 {Proceedings of Thirty Eighth Conference on Learning Theory},
  pages = 	 {4309--4358},
  year = 	 {2025},
  volume = 	 {291},
  series = 	 {Proceedings of Machine Learning Research},
  month = 	 {7},
  publisher =    {PMLR},
  eprinttype = {arXiv},
  eprint = {2408.12385}
}

@misc{wang2026nearlytightlowerbounds,
      title={A Unified Complexity Framework for Quantum Property Testing},
      author={Qisheng Wang},
      year={2026},
      eprint={2608.02600},
      archivePrefix={arXiv},
      primaryClass={quant-ph},
      url={https://arxiv.org/abs/2608.02600}, 
}

@article{AISW20,
  author = {Acharya, Jayadev and Issa, Ibrahim and Shende, Nirmal V. and Wagner, Aaron B.},
  title = {Estimating Quantum Entropy},
  journal = {IEEE Journal on Selected Areas in Information Theory},
  volume = {1},
  number = {2},
  pages = {454--468},
  year = {2020},
  doi = {10.1109/JSAIT.2020.3015235},
  eprinttype = {arXiv},
  eprint = {1711.00814}
}

@misc{GW26,
  author = {Gao, Minbo and Wang, Qisheng},
  title = {Breaking the Quadratic Barrier for von {Neumann} Entropy Estimation},
  year = {2026},
  eprinttype = {arXiv},
  eprint = {2608.11151}
}

@incollection{DeVore76,
  author = {DeVore, Ronald A.},
  title = {Degree of approximation},
  booktitle = {Approximation Theory {II}},
  editor = {Lorentz, G. G. and Chui, C. K. and Schumaker, L. L.},
  publisher = {Academic Press},
  location = {New York},
  pages = {117--162},
  year = {1976}
}

@InProceedings{HJW18,
  author = {Han, Yanjun and Jiao, Jiantao and Weissman, Tsachy},
  title = {Local moment matching: A unified methodology for symmetric functional estimation and distribution estimation under {Wasserstein} distance},
  booktitle = {Proceedings of the 31st Conference On Learning Theory},
  pages = {3189--3221},
  year = {2018},
  volume = {75},
  series = {Proceedings of Machine Learning Research},
  publisher = {PMLR},
  eprinttype = {arXiv},
  eprint = {1802.08405}
}

@misc{PSTW26,
  author = {Pelecanos, A. and Spilecki, J. and Tang, E. and Wright, J.},
  title = {The {Keyl--Werner} algorithm is not optimal for spectrum estimation},
  year = {2026},
  eprinttype = {arXiv},
  eprint = {2607.27117}
}

@inproceedings{OW16,
  author = {O'Donnell, Ryan and Wright, John},
  title = {Efficient quantum tomography},
  booktitle = {Proceedings of the Forty-Eighth Annual ACM Symposium on Theory of Computing},
  pages = {899--912},
  year = {2016},
  doi = {10.1145/2897518.2897544},
  eprinttype = {arXiv},
  eprint = {1508.01907}
}

@article{OW21,
  author = {O'Donnell, Ryan and Wright, John},
  title = {Quantum spectrum testing},
  journal = {Communications in Mathematical Physics},
  volume = {387},
  number = {1},
  pages = {1--75},
  year = {2021},
  doi = {10.1007/s00220-021-04180-1},
  eprinttype = {arXiv},
  eprint = {1501.05028}
}

@article{HHJWY17,
  author = {Haah, Jeongwan and Harrow, Aram W. and Ji, Zhengfeng and Wu, Xiaodi and Yu, Nengkun},
  title = {Sample-optimal tomography of quantum states},
  journal = {IEEE Transactions on Information Theory},
  volume = {63},
  number = {9},
  pages = {5628--5641},
  year = {2017},
  doi = {10.1109/TIT.2017.2719044},
  eprinttype = {arXiv},
  eprint = {1508.01797}
}

@inproceedings{CHLLS23,
  author = {Chen, Sitan and Huang, Brice and Li, Jerry and Liu, Allen and Sellke, Mark},
  title = {When does adaptivity help for quantum state learning?},
  booktitle = {Proceedings of the 64th Annual IEEE Symposium on Foundations of Computer Science},
  pages = {391--404},
  year = {2023},
  eprinttype = {arXiv},
  eprint = {2206.05265}
}

@misc{FOW26,
  author = {Fanizza, Marco and O'Donnell, Ryan and Wadhwa, Chirag},
  title = {Spectrum estimation is almost as hard as tomography},
  year = {2026},
  eprinttype = {arXiv},
  eprint = {2607.29680}
}

@misc{LJ26,
  author = {Lee, Gye Jin and Jo, Sunghyeon},
  title = {The sample complexity of fidelity estimation to a known rank-$r$ reference state is $\widetilde{\Theta}(r^2/\varepsilon^2)$},
  year = {2026},
  eprinttype = {arXiv},
  eprint = {2608.01770}
}

@phdthesis{Wright16,
  author = {Wright, John},
  title = {How to learn a quantum state},
  school = {Carnegie Mellon University},
  year = {2016},
  month = {5},
  number = {CMU-CS-16-108},
  url = {https://people.eecs.berkeley.edu/~jswright/papers/thesis.pdf}
}

@article{BHACRG18,
  author = {Beverland, Michael E. and Haah, Jeongwan and Alagic, Gorjan and Campbell, Gretchen K. and Rey, Ana Maria and Gorshkov, Alexey V.},
  title = {Spectrum estimation of density operators with alkaline-earth atoms},
  journal = {Physical Review Letters},
  volume = {120},
  number = {2},
  eid = {025301},
  year = {2018},
  doi = {10.1103/PhysRevLett.120.025301},
  eprinttype = {arXiv},
  eprint = {1608.02045}
}

@article{BTBLR13,
  author = {Boguslawski, Katharina and Tecmer, Pawel and Barcza, Gergely and Legeza, {\"O}rs and Reiher, Markus},
  title = {Orbital entanglement in bond-formation processes},
  journal = {Journal of Chemical Theory and Computation},
  volume = {9},
  number = {7},
  pages = {2959--2973},
  year = {2013},
  doi = {10.1021/ct400247p},
  eprinttype = {arXiv},
  eprint = {1303.7207}
}

@article{VV17,
  author = {Valiant, Gregory and Valiant, Paul},
  title = {Estimating the unseen: Improved estimators for entropy and other properties},
  journal = {Journal of the ACM},
  volume = {64},
  number = {6},
  eid = {37},
  pagetotal = {41},
  year = {2017},
  doi = {10.1145/3125643},
  url = {https://www.cs.purdue.edu/homes/pvaliant/unseen-jacm.pdf}
}

@article{HM02,
  author = {Hayashi, Masahito and Matsumoto, Keiji},
  title = {Quantum universal variable-length source coding},
  journal = {Physical Review A},
  volume = {66},
  number = {2},
  eid = {022311},
  year = {2002},
  doi = {10.1103/PhysRevA.66.022311},
  eprinttype = {arXiv},
  eprint = {quant-ph/0202001}
}

@book{OD14,
  author = {O'Donnell, Ryan},
  title = {Analysis of {Boolean} Functions},
  publisher = {Cambridge University Press},
  year = {2014},
  doi = {10.1017/CBO9781139814782}
}

@book{vdV98,
  author = {van der Vaart, A. W.},
  title = {Asymptotic Statistics},
  publisher = {Cambridge University Press},
  year = {1998},
  doi = {10.1017/CBO9780511802256}
}

@article{GB10,
  author = {Gu\c{t}\u{a}, M\u{a}d\u{a}lin and Butucea, Cristina},
  title = {Quantum {U}-statistics},
  journal = {Journal of Mathematical Physics},
  volume = {51},
  number = {10},
  eid = {102202},
  year = {2010},
  doi = {10.1063/1.3476776},
  eprinttype = {arXiv},
  eprint = {1004.2452}
}

@misc{PFMO25,
  author = {De Palma, Giacomo and Fanizza, Marco and Mowry, Connor and O'Donnell, Ryan},
  title = {Non-iid hypothesis testing: from classical to quantum},
  year = {2025},
  eprinttype = {arXiv},
  eprint = {2510.06147}
}

@book{rivlin1990chebyshev,
  title     = {Chebyshev Polynomials: From Approximation Theory to Algebra and Number Theory},
  author    = {Rivlin, Theodore J.},
  series    = {Pure and Applied Mathematics: A Wiley Series of Texts, Monographs and Tracts},
  volume    = {10},
  edition   = {2nd},
  year      = {1990},
  publisher = {Wiley},
  address   = {New York},
  isbn      = {9780471628965}
}

@inproceedings{BKM22,
author = {Braverman, Vladimir and Krishnan, Aditya and Musco, Christopher},
title = {Sublinear time spectral density estimation},
year = {2022},
isbn = {9781450392648},
publisher = {Association for Computing Machinery},
url = {https://doi.org/10.1145/3519935.3520009},
doi = {10.1145/3519935.3520009},
booktitle = {Proceedings of the 54th Annual ACM SIGACT Symposium on Theory of Computing},
pages = {1144–1157},
numpages = {14},
}

@Book{Stanley1999,
author={Stanley, Richard P.},
title={Enumerative Combinatorics},
series={Cambridge Studies in Advanced Mathematics},
year={1999},
publisher={Cambridge University Press},
address={Cambridge},
volume={2},
doi={10.1017/CBO9780511609589},
url={https://www.cambridge.org/product/D8DDDFF7E8EBF0BCFE99F5E6918CE2A8},
url={https://doi.org/10.1017/CBO9780511609589}
}

@book{Jackson:1930,
	author = {Dunham Jackson},
	publisher = {American Mathematical Society},
	series = {Colloquium Publications},
	title = {The Theory of Approximation},
	volume = {11},
	year = {1930}}

@misc{LT26,
      title={Random dimension reduction and learning symmetric properties of quantum states}, 
      author={Angus Lowe and Xinyu Tan},
      year={2026},
      eprint={2606.23592},
      archivePrefix={arXiv},
      primaryClass={quant-ph},
      url={https://arxiv.org/abs/2606.23592}, 
}

@inproceedings{VV11a,
author = {Valiant, Gregory and Valiant, Paul},
title = {Estimating the unseen: an n/log(n)-sample estimator for entropy and support size, shown optimal via new CLTs},
year = {2011},
isbn = {9781450306911},
publisher = {Association for Computing Machinery},
address = {New York, NY, USA},
url = {https://doi.org/10.1145/1993636.1993727},
doi = {10.1145/1993636.1993727},
booktitle = {Proceedings of the Forty-Third Annual ACM Symposium on Theory of Computing},
pages = {685–694},
numpages = {10},
location = {San Jose, California, USA},
series = {STOC '11}
}

@misc{UNWT25,
      title={Quantum algorithms for Uhlmann transformation}, 
      author={Takeru Utsumi and Yoshifumi Nakata and Qisheng Wang and Ryuji Takagi},
      year={2025},
      eprint={2509.03619},
      archivePrefix={arXiv},
      primaryClass={quant-ph},
      url={https://arxiv.org/abs/2509.03619}, 
}

@misc{GP22,
      title={Improved Quantum Algorithms for Fidelity Estimation}, 
      author={András Gilyén and Alexander Poremba},
      year={2022},
      eprint={2203.15993},
      archivePrefix={arXiv},
      primaryClass={quant-ph},
      url={https://arxiv.org/abs/2203.15993}, 
}

\appendix
\section{Chebyshev derivative bounds for the variance analysis}
\label{app:chebyshev-derivatives}

We record the derivative bounds used in
\Cref{cor:explicit-chebyshev-variance}.

\begin{lemma}[Chebyshev derivative bounds]
\label{lem:cheb-derivatives}
For $1\leq r\leq k$,
\begin{equation*}
 \frac{\|\phi_k^{(r)}\|_{\infty,[0,L]}}{r!}
 \leq\frac{(4k^2/L)^r}{(2r)!},
 \qquad
 \frac{\|\psi_k^{(r)}\|_{\infty,[0,L]}}{r!}
 \leq\frac{(2k/L)^r}{r!}.
\end{equation*}
\end{lemma}

\begin{proof}
The derivatives of $T_k$ satisfy (see \cite[Theorem 2.24]{rivlin1990chebyshev})
\begin{equation*}
 \max_{|t|\leq1}|T_k^{(r)}(t)|=T_k^{(r)}(1),
 \qquad
 \frac{T_k^{(r)}(1)}{r!}
 =\frac{2^r}{(2r)!}\prod_{j=0}^{r-1}(k^2-j^2).
\end{equation*}
Applying the rescaling along with the fact that $\prod_{j=0}^{r-1}(k^2-j^2)\leq k^{2r}$
gives the bound for $\phi_k$.

For the interior bound, fix $|t|\leq1/2$ and let
$P_t(u)=T_k(t+u/2)$.
This polynomial has degree $k$ and supremum norm at most one on
$[-1,1]$.  Differentiating $P_t$ $r$ times and evaluating at $u=0$ gives us that $T_k^{(r)}(t) = 2^r P_t^{(r)}(0) $. Therefore, our aim is to bound $ \abs{P^{(r)}(0)}$ for any polynomial $P$ of degree $k$ and supremum norm at most one on
$[-1,1]$. Using \cite[Remark 2]{rivlin1990chebyshev}, we get that  $\abs{P^{(r)}(0)} \leq (2k)^r$. Applying the rescaling then gives the bound for $\psi_k$.
\end{proof}

\section{Jackson approximation and rescaling}
\label{sec:jackson-approximation}

\begin{proof}[Proof of \Cref{lem:jackson-endpoint-improvement}]
We use the Jackson construction from
\cite[Appendix~C]{BKM22}. Set $m=\lfloor K/2\rfloor+1$,
and let
\[
 b(u)=\left(\frac{\sin(mu/2)}{\sin(u/2)}\right)^4 = \sum_{k=-2m+2}^{2m-2} \widehat b_k e^{i k u}
\]
be the kernel from \cite[Definition~C.4]{BKM22}, with
Fourier coefficients $\widehat b_k$. Following the proof of
\cite[Theorem~C.6]{BKM22}, let $\tilde h$ be the periodic convolution of $h$  with the normalized Jackson kernel $b/(2 \pi \widehat b_0)$, i.e., 
\[
 h(\theta)=f(\cos\theta),
 \qquad
 \widetilde h(\theta)
 =\frac{1}{2\pi\widehat b_0}
   \int_{-\pi}^{\pi}b(u)h(\theta-u)\diff u.
\]

Since $T_k(\cos\theta)=\cos(k\theta)$, the Chebyshev expansion
of $f$ gives
\[
 h(\theta)=f(\cos\theta)
 =\gamma_0+\sum_{k\geq1}\gamma_k\cos(k\theta).
\]
Convolution with the normalized Jackson kernel multiplies
the $k$th cosine coefficient by $\widehat b_k/\widehat b_0$.
Since $\widehat b_k=0$ for $k>2m-2$, we obtain
\[
 \widetilde h(\theta)
 =\gamma_0+\sum_{k=1}^{2m-2}
 \frac{\widehat b_k}{\widehat b_0}\gamma_k\cos(k\theta).
\]
Thus, we have 
\[
 Q_f(t)
 :=\gamma_0+\sum_{k=1}^{2m-2}
 \frac{\widehat b_k}{\widehat b_0}\gamma_kT_k(t), \qquad Q_f(\cos\theta)
 =\gamma_0+\sum_{k=1}^{2m-2}
 \frac{\widehat b_k}{\widehat b_0}\gamma_k\cos(k\theta)
 =\widetilde h(\theta).
\]
In particular, $Q_f$ has degree at most $2m-2\leq K$, and by  \cite[Theorem~C.6]{BKM22}, we have 
\[
 Q_f(t)=\gamma_0+\sum_{k=1}^{K}a_kT_k(t),
 \qquad
 a_k=\eta_k\gamma_k,
 \qquad
 \eta_k=\frac{\widehat b_k}{\widehat b_0}\in[0,1],
\]
with $\eta_k=0$ for $k>2m-2$. This establishes the
claimed coefficient form.

\paragraph{Kernel moments.} Since $\widehat b_k$ are the Fourier coefficients, we have that $\widehat b_0 = \frac{1}{2 \pi}\int_{-\pi}^{\pi} b(u) \diff u$ and $\widehat b_1 = \frac{1}{2 \pi}\int_{-\pi}^{\pi} b(u) \cos(u)  \diff u$.
The coefficient formula in \cite[Eq.~(16)]{BKM22} gives
\[
 \widehat b_0=\frac{2m^3+m}{3},
 \qquad
 \widehat b_0-\widehat b_1=m.
\]
Define
\[
 M_j=\frac{1}{2\pi\widehat b_0}
       \int_{-\pi}^{\pi}|u|^j b(u)\diff u,
 \qquad j=1,2.
\]
Since $u^2\leq(\pi^2/2)(1-\cos u)$ on $[-\pi,\pi]$,
\[
 M_2 = \frac{1}{2\pi\widehat b_0}
       \int_{-\pi}^{\pi} u^2 b(u)\diff u
       \leq \frac{\pi^2/2}{2\pi\widehat b_0}\int_{-\pi}^{\pi} (1-\cos (u)) b(u) \diff u
 \leq\frac{\pi^2}{2} 
       \left(\frac{\widehat b_0 - \widehat b_1}{\widehat b_0}\right)
 =\frac{\pi^2m}{2\widehat b_0}
 \leq\frac{C}{K^2}.
\]
The normalized kernel $b/(2\pi\widehat b_0)$ is nonnegative
and integrates to one, so Cauchy--Schwarz also gives
$M_1\leq\sqrt{M_2}\leq C/K$.

\paragraph{The pointwise estimate.}
Fix $\theta\in[0,\pi]$. We refine the convolution estimate
in the proof of \cite[Theorem~C.5]{BKM22} by using the fact that $f$ is a $1$-Lipschitz function, 
\[
 |h(\theta-u)-h(\theta)|
 \leq|\cos(\theta-u)-\cos\theta|
 \leq\sin\theta\,|u|+\frac{u^2}{2}.
\]
Integrating against the normalized kernel yields
\[
 |\widetilde h(\theta)-h(\theta)|
 \leq\sin\theta\,M_1+\frac{M_2}{2}
 \leq C\left(\frac{\sin\theta}{K}+\frac1{K^2}\right).
\]
Substituting $t=\cos\theta$ proves
\[
 |Q_f(t)-f(t)|
 \leq C\left(\frac{\sqrt{1-t^2}}{K}+\frac1{K^2}\right).
\]

\paragraph{The anchored estimate.}
Assume now that $f(-1)=0$, and set
$P_f(t)=Q_f(t)-Q_f(-1)$. Then $h(\pi)=0$ and
\[
 P_f(\cos\theta)
 =\widetilde h(\theta)-\widetilde h(\pi).
\]
Write $\delta=1+\cos\theta$, so that
$\sin\theta\leq\sqrt{2\delta}$. The preceding bound contains an additive \(K^{-2}\) term, so we need a second estimate that vanishes as \(\delta\to 0\).
The preceding estimate at $\theta$ and at $\pi$ gives
\[
 |P_f(\cos\theta)-f(\cos\theta)|
 \leq\sin\theta\,M_1+M_2
 \leq C\left(\frac{\sqrt\delta}{K}+\frac1{K^2}\right).
\]

Near the anchored endpoint, we instead compare the two
convolution integrands directly. The Lipschitz property gives
\begin{align*}
 |h(\theta-u)-h(\pi-u)|
 &\leq|\cos(\theta-u)+\cos u|\\
 &\leq(1+\cos\theta)|\cos u|
       +\sin\theta\,|\sin u|\\
 &\leq\delta+\sin\theta\,|u|.
\end{align*}
Thus
\[
 |P_f(\cos\theta)|\leq\delta+\sin\theta\,M_1.
\]
Since $|f(\cos\theta)|\leq\delta$, we obtain the second bound
\[
 |P_f(\cos\theta)-f(\cos\theta)|
 \leq2\delta+C\frac{\sqrt\delta}{K}.
\]

For $\delta\geq K^{-2}$, use the first bound and
$K^{-2}\leq\sqrt\delta/K$.
For $\delta\leq K^{-2}$, use the second bound and
$\delta\leq\sqrt\delta/K$.
Both cases give
\begin{align*}
    |P_f(t)-f(t)|\leq C\frac{\sqrt{1+t}}{K},
 \qquad t\in[-1,1]. & \qedhere
\end{align*}
\end{proof}

\section{Proof of the trace derivative formula}
\label{app:trace-derivatives}

\tracederivatives*

\begin{proof}
We first obtain the formula by induction, then bound its coefficients
using Rolle's theorem. Recall that 
\[
D_p^{(r)}(H_1,\ldots,H_r)
:=
\left.
\frac{\partial^r}{\partial t_r\cdots\partial t_1}
\tr p\!\left(\sigma+\sum_{j=1}^r t_jH_j\right)
\right|_{t_1=\cdots=t_r=0}.
\]

\paragraph{The product-rule induction.}
Start with $p(x)=x^n$ and write
$A=\sigma+\sum_\ell t_\ell H_\ell$, with one variable for each
direction under consideration. Thus $\partial A/\partial t_\ell=H_\ell$.
Differentiating one factor at a time gives
\[
\frac{\partial}{\partial t_1}\tr(A^n)
=\sum_{a+b=n-1}\tr(A^aH_1A^b)
=n\tr(H_1A^{n-1}).
\]
The last equality uses cyclicity of the trace to move $H_1$ to the
front. Further derivatives act only on the powers of $A$, so $H_1$
stays first. All exponents below are nonnegative integers, and $A^0=I$.

For $n\ge r$, we claim that repeated differentiation gives
\begin{equation}\label{eq:word-formula}
\frac{\partial^r}{\partial t_r\cdots\partial t_1}\tr(A^n)
=n\sum_{\substack{\pi\in S_r\\\pi(1)=1}}
\sum_{a_1+\cdots+a_r=n-r}
\tr\!\left(
H_1A^{a_1}H_{\pi(2)}A^{a_2}\cdots H_{\pi(r)}A^{a_r}
\right).
\end{equation}
The case $r=1$ was just proved. Assume the formula holds for $r$.
To differentiate once more, the product rule acts on each block:
\[
\frac{\partial}{\partial t_{r+1}}A^{a_j}
=\sum_{b+c=a_j-1}A^bH_{r+1}A^c.
\]

By the product rule, we get that 
\begin{align*}
&\frac{\partial^{r+1}}{\partial t_{r+1}\cdots\partial t_1}\tr(A^n)
\\
&\quad=n\sum_{\substack{\pi\in S_r\\\pi(1)=1}}
\sum_{a_1+\cdots+a_r=n-r}\sum_{j=1}^r\sum_{b+c=a_j-1}
\tr\!\left(
\left[\prod_{\ell=1}^{j-1}H_{\pi(\ell)}A^{a_\ell}\right]
H_{\pi(j)}A^bH_{r+1}A^c
\left[\prod_{\ell=j+1}^{r}H_{\pi(\ell)}A^{a_\ell}\right]
\right).
\end{align*}

For each $j$, substitute $a_j=b+c+1$. Rename the $r+1$ exponents as
$\alpha_1,\ldots,\alpha_{r+1}$. Moreover, $\pi$ specifies the order of the old $H$'s and $j$ specifies
where to insert $H_{r+1}$. Together they give each new order $\rho$
with $\rho(1)=1$. Combining these two relabelings gives
\begin{equation*}
\frac{\partial^{r+1}}{\partial t_{r+1}\cdots\partial t_1}\tr(A^n)
={} n\sum_{\substack{\rho\in S_{r+1}\\\rho(1)=1}}
\sum_{\alpha_1+\cdots+\alpha_{r+1}=n-r-1}
\tr\!\left(
H_1A^{\alpha_1}H_{\rho(2)}A^{\alpha_2}\cdots
H_{\rho(r+1)}A^{\alpha_{r+1}}
\right).
\end{equation*}
This is \cref{eq:word-formula} with $r+1$ in place of $r$.
If $n<r$, the derivative is zero because $A^n$ has total degree at
most $n$ in the variables $t_\ell$.

At $t=0$, we have $A=\sigma$. Fix an order $\pi$ and a choice of
exponents $a_1,\ldots,a_r$. Ordinary matrix multiplication and the
definition of trace give
\[
\tr\!\left(H_1\sigma^{a_1}\cdots H_{\pi(r)}\sigma^{a_r}\right)
=\sum_{i_1,\ldots,i_r=1}^d
\left[\prod_{j=1}^r(H_{\pi(j)})_{i_ji_{j+1}}\right]
\left[\prod_{j=1}^r\lambda_{i_{j+1}}^{a_j}\right],
\qquad i_{r+1}=i_1.
\]
Here we used
$(H_{\pi(j)}\sigma^{a_j})_{uv}
=(H_{\pi(j)})_{uv}\lambda_v^{a_j}$, since $\sigma$ is diagonal.
Note that the products in brackets consist of scalar entries. Summing over the exponents $a_1,\dots,a_r$, we get that
\begin{equation*}
\sum_{a_1+\cdots+a_r=n-r}
\tr\!\left(H_1\sigma^{a_1}\cdots H_{\pi(r)}\sigma^{a_r}\right)
=
\sum_{i_1,\ldots,i_r=1}^d
\left[\prod_{j=1}^r(H_{\pi(j)})_{i_ji_{j+1}}\right]
\left[
\sum_{a_1+\cdots+a_r=n-r}
\lambda_{i_2}^{a_1}\lambda_{i_3}^{a_2}\cdots\lambda_{i_1}^{a_r}
\right],
\end{equation*}
since for each fixed tuple $(i_1,\ldots,i_r)$, the entire product of
$H$-entries is independent of the $a_j$'s. For $s\ge0$, let
\[
h_s(x_1,\ldots,x_r)
:=\sum_{\substack{a_1+\cdots+a_r=s\\a_j\ge0}}
x_1^{a_1}\cdots x_r^{a_r},
\]
with $h_s=0$ for $s<0$.

Consequently, for $p(x)=\sum_{n=0}^k c_nx^n$, the formula obtained
directly by differentiation is
\[
D_p^{(r)}(H_1,\ldots,H_r)
= \frac{1}{(r-1)!} \sum_{\substack{\pi\in S_r\\\pi(1)=1}}
\sum_{i_1,\ldots,i_r=1}^d
w_{i_1,\ldots,i_r}
\prod_{j=1}^r(H_{\pi(j)})_{i_ji_{j+1}},
\]
where
\begin{equation}\label{eq:w_def}
    w_{i_1,\ldots,i_r}
:=(r-1)! \sum_{n=r}^k n c_n
h_{n-r}(\lambda_{i_1},\ldots,\lambda_{i_r}).
\end{equation}

\paragraph{The coefficient bound.}
For $r=1$, it can be checked from \cref{eq:word-formula} that $w_i=p'(\lambda_i)$, so the bound
is immediate. Suppose $r\ge2$. Fix an index tuple, and abbreviate
$x_j=\lambda_{i_j}$ and $w=w_{i_1,\ldots,i_r}$.
First suppose the values $x_1,\ldots,x_r$ are pairwise distinct. Next, we aim to express $w$ in terms of the values of $p'$ at these
points. The reason to look for $p'$ is that the formula for $w$ contains
the coefficients $n c_n$, and $p'(x)=\sum_{n=1}^k n c_nx^{n-1}$.
By induction (or see \cite[Exercise 7.4]{Stanley1999}), we have
\begin{equation}\label{eq:scalar-identity}
\sum_{j=1}^r\frac{x_j^m}{\prod_{\ell\ne j}(x_j-x_\ell)}
=h_{m-r+1}(x_1,\ldots,x_r),\qquad m\ge r-1.
\end{equation}

For the lower powers $0\le m<r-1$, Lagrange interpolation gives
\[
x^m=\sum_{j=1}^r x_j^m
\prod_{\ell\ne j}\frac{x-x_\ell}{x_j-x_\ell}.
\]
Since $m < r-1$, the coefficient of $x^{r-1}$ on the left is zero. Comparing it with
the coefficient on the right shows that
\begin{equation}\label{eq:lower-powers}
\sum_{j=1}^r\frac{x_j^m}{\prod_{\ell\ne j}(x_j-x_\ell)}=0,
\qquad 0\le m<r-1.
\end{equation}

Substitute \cref{eq:scalar-identity}, with $m=n-1$, into the formula
for $w$. \cref{eq:lower-powers} lets us extend the sum down
to $n=1$, since each added inner sum is zero. Interchanging the finite
sums then reveals $p'$:
\begin{align}
\frac{w}{(r-1)!}
&=\sum_{n=r}^k n c_n h_{n-r}(x_1,\ldots,x_r)  &\left( \text{From \cref{eq:w_def}} \right) \nonumber
\\
&=\sum_{n=r}^k n c_n
\sum_{j=1}^r\frac{x_j^{n-1}}{\prod_{\ell\ne j}(x_j-x_\ell)} & \left( \text{From \cref{eq:scalar-identity}} \right) \nonumber
\\
&=\sum_{n=1}^k n c_n
\sum_{j=1}^r\frac{x_j^{n-1}}{\prod_{\ell\ne j}(x_j-x_\ell)} & \left( \text{From \cref{eq:lower-powers}} \right) \nonumber
\\
&=\sum_{j=1}^r
\frac{\sum_{n=1}^k n c_nx_j^{n-1}}{\prod_{\ell\ne j}(x_j-x_\ell)}
=\sum_{j=1}^r\frac{p'(x_j)}{\prod_{\ell\ne j}(x_j-x_\ell)}. \label{eq:w_expression}
\end{align}

We now want to turn this expression for $w$ into a derivative bound.
For this, we seek a polynomial $Q$ that agrees with $p'$ at all $r$
points and satisfies $Q^{(r-1)}=w$. Towards that, let $Q$ be the Lagrange interpolating polynomial of 
$p'$, as follows:
\begin{equation} \label{eq:qx_def}
    Q(x) := \sum_{j=1}^r p'(x_j)\prod_{\ell\ne j}\frac{x-x_\ell}{x_j-x_\ell}.
\end{equation}

Combining \cref{eq:w_expression} and \cref{eq:qx_def}, we get that 
\[
Q(x)=\frac{w}{(r-1)!}x^{r-1}
+\text{terms of degree at most }r-2.
\]
Differentiating $r-1$ times kills every lower-degree term, while
the derivative of $x^{r-1}$ is $(r-1)!$. Thus
\[
Q^{(r-1)}(x)=\frac{w}{(r-1)!}(r-1)!=w.
\]

Now set $F=p'-Q$. Since $Q(x_j)=p'(x_j)$, the polynomial $F$ has
$r$ distinct zeros (since we assumed the $x_j$'s to be distinct). Rolle's theorem gives at least $r-1$ zeros of
$F'$, then at least $r-2$ zeros of $F''$, and so on. After $r-1$
steps, there is a point $\xi$ between the smallest and largest $x_j$
with $F^{(r-1)}(\xi)=0$. At this point,
\[
0=F^{(r-1)}(\xi)
=(p')^{(r-1)}(\xi)-Q^{(r-1)}(\xi)
=p^{(r)}(\xi)-w.
\]
Consequently $w=p^{(r)}(\xi)$, and $\xi\in[0,L]$ gives
\[
{|w|\le\norm{p^{(r)}}_{\infty,[0,L]}}.
\]

If some \(x_j\) coincide, perturb them to distinct points in \([0,L]\) and pass to the limit; the bound is preserved because \(w\) is a polynomial in \(x_1,\ldots,x_r\), hence continuous. 
\end{proof}

\end{document}